\documentclass[10pt, onecolumn]{IEEEtran}
\IEEEoverridecommandlockouts \def\withcolors{1} \def\withnotes{1}
\def\journal{0}			
\usepackage[normalem]{ulem}

\PassOptionsToPackage{hyphens}{url}
\usepackage[colorlinks=true,citecolor=blue,linkcolor=blue,linktocpage]{hyperref}
\usepackage{makecell}
\usepackage{cite}
\usepackage{graphicx}
\usepackage{caption}
\usepackage[labelformat=simple]{subcaption}

\usepackage{epstopdf}
\usepackage[usenames,dvipsnames]{color}
\usepackage{epsfig}
\usepackage{amssymb}
\usepackage{amsmath}
\usepackage{amsthm}
\usepackage{latexsym}
\usepackage{setspace}
\usepackage{bbm}
\usepackage[top=1in, bottom=1in, left=1in, right=1in]{geometry}
\usepackage{float}
\usepackage{dsfont}
\usepackage[ruled,vlined]{algorithm2e}
\usepackage{xspace}
\usepackage{enumerate}

\usepackage{comment}

\newcommand\numberthis{\addtocounter{equation}{1}\tag{\theequation}}

\newcommand{\dP}{\mathrm{P}}

\newcommand{\bPP}[1]{\dP_{#1}}

\newcommand{\bPr}[1]{{\mathrm{Pr}}\left(#1\right)}

\newcommand{\bEE}[1]{{\mathbb{E}}\left[#1\right]}

\newcommand{\cN}{{\mathcal N}}
\newcommand{\mN}{{\mathbbm N}}
\newcommand{\mR}{{\mathbbm R}}

\newcommand{\cS}{{\mathcal S}}

\makeatletter
\newtheorem*{rep@theorem}{\rep@title}
\newcommand{\newreptheorem}[2]{%
\newenvironment{rep#1}[1]{%
 \def\rep@title{#2 \ref{##1}}%
 \begin{rep@theorem}}%
 {\end{rep@theorem}}}
\makeatother

\newtheorem{theorem}{Theorem}

\newtheorem*{corollary*}{Corollary}
\newtheorem{assumption}{Assumption}

\newtheorem*{assumptions*}{Assumptions}
\newtheorem{lemma}{Lemma}
\newtheorem*{lemma*}{Lemma}

\newtheorem{lemma-app}{Lemma}[section]

\newreptheorem{theorem}{Theorem}
\newreptheorem{lemma}{Lemma}

\theoremstyle{remark}
\newtheorem{remark}{Remark}
\newtheorem*{remark*}{Remark}
\newtheorem*{remarks*}{Remarks}

\theoremstyle{definition}
\newtheorem{definition}{Definition}

\newcommand{\ed}{\stackrel{{\rm def}}{=}}

\def\undertilde#1{\mathord{\vtop{\ialign{##\crcr
$\hfil\displaystyle{#1}\hfil$\crcr\noalign{\kern1.5pt\nointerlineskip}
$\hfil\tilde{}\hfil$\crcr\noalign{\kern1.5pt}}}}}

\newcommand{\ep}{\varepsilon}

\newcommand{\tht}{\theta}

\allowdisplaybreaks

\newcommand{\ns}{n}  
\newcommand{\dmn}{d} 
\newcommand{\sps}{k}  
\newcommand{\nm}{m}  
\newcommand{\normzero}[1]{\left\|{#1}\right\|_0}  

\newcommand{\bx}{X}				
\newcommand{\by}{Y}				
\newcommand{\bw}{W}				
\newcommand{\bphi}{\Phi}		

\newcommand{\sbg}{\text{subG}}
\newcommand{\sbx}{\text{subexp}}
\newcommand{\lmin}{\lambda_{\text{min}}}
\newcommand{\lmax}{\lambda_{\text{max}}}

\newcommand{\se}{{\mathbb{E}}}

\ifnum\journal=1
	\includecomment{onlyjournal}
	\excludecomment{onlyarxiv}
\else
	\includecomment{onlyarxiv}
	\excludecomment{onlyjournal}
\fi

\ifnum\withcolors=1
 
  \newcommand{\ccolor}[1]{{\color{RubineRed}#1}} 
  \newcommand{\hcolor}[1]{{\color{Orange}#1}} 
\else

  \newcommand{\acolor}[1]{{#1}}
  \newcommand{\ccolor}[1]{{#1}}
  \newcommand{\hcolor}[1]{{#1}}
\fi

\ifnum\withnotes=1

  \newcommand{\lnote}[1]{\par\acolor{\textbf{L: }\sf #1}} 
  \newcommand{\cnote}[1]{\par\ccolor{\textbf{C: }\sf #1}} 
  \newcommand{\hnote}[1]{\par\hcolor{\textbf{H: }\sf #1}} 

\else

  \newcommand{\lnote}[1]{}
  \newcommand{\cnote}[1]{}
  \newcommand{\hnote}[1]{}

\fi
\newcommand{\ignore}[1]{\leavevmode\unskip} 

\begin{document}

\title{Sample-Measurement Tradeoff in Support Recovery under a
 Subgaussian Prior}

\author{ \IEEEauthorblockN{Lekshmi Ramesh, Chandra R. Murthy, and Himanshu Tyagi}}

\maketitle {\renewcommand{\thefootnote}{}\footnotetext{
		\noindent The authors are with the Department of
                Electrical Communication Engineering, Indian Institute
                of Science, Bangalore 560012, India.  Email:
                \{lekshmi, cmurthy, htyagi\}@iisc.ac.in.} }
	
	{\renewcommand{\thefootnote}{}\footnotetext{
			\noindent \thanks{This work was financially
                          supported by a research fellowship from the
                          Ministry of Electronics and Information
                          Technology, Govt. of India.}} }
                          
  	{\renewcommand{\thefootnote}{}\footnotetext{
          	\noindent This work appeared in part in ISIT 2019 \cite{Ramesh_ISIT_2019}. }
	
\begin{abstract}
	Data samples from $\mathbb{R}^{d}$ with a common support of size
        $k$ are accessed through $m$ random linear projections (measurements)
        per sample.  It is well-known that roughly $k$ measurements
        from a single sample are sufficient to recover the support.
        In the multiple sample setting, do $k$ \emph{overall} measurements
        still suffice when only $m$ measurements \emph{per sample} are
        allowed, with $m<k$?  
        We answer this question {in the negative} by
        considering a generative model setting with independent
        samples drawn from a subgaussian prior.  We show that 
        $n=\Theta((k^2/m^2)\cdot\log\sps(d-k))$ samples are
        necessary and sufficient to recover the support exactly. In
        turn, this shows that \emph{when $m<k$, $k$ overall measurements are insufficient
        for support recovery}; instead we need about
      $m$ measurements each from $k^{2}/m^2$ samples, i.e.,  $k^{2}/m$ overall measurements are necessary.

\end{abstract}




\section{Introduction}\label{sec:introduction}

  Consider $\ns$ vectors (hereafter, referred to as samples) $X_{1},\ldots,X_{\ns}$ taking
  values in $\mR^{\dmn}$ and with a common support $S \subset [\dmn]$ of cardinality $\sps \ll \dmn$. It is easy to find
this common support by simply checking each coordinate of any single sample $X_i$, but this
will require $\dmn$ measurements. 
It is now well-known from the theory of compressed sensing that a much smaller number of measurements suffices. In particular, given random linear measurements $Y_i=\bphi X_i \in \mR^{\nm}$ of any \emph{single} sample $X_i$, one can recover the support with $\nm=O(\sps \log (\dmn-\sps))$ measurements \cite{Wainwright_TIT_2009}.
However, considering only one sample requires $\nm>\sps$
measurements per sample. It also does not exploit the fact that the samples share a common
support, which can be useful when the measurement process can be parallelized across samples.  
A natural question that
arises then is whether, in the multiple sample setting,
we can reduce the number of measurements
required per sample, and still recover the support correctly.
That is, can we can still recover the unknown support with $\sps$ \emph{overall} measurements (i.e., would $\ns \nm \le \sps$ suffice)?

Surprisingly, this tradeoff between $\nm$ and $\ns$ has not been explored in the literature. In this work, we provide an explicit characterization of this tradeoff.
We examine this question in a natural Bayesian setting and show that when $\nm < \sps$, we will need at least $\sps^2/\nm$
overall measurements. Thus, in  contrast with the $\nm > \sps$
regime where $\sps$ overall measurements from a single sample suffice, the support can still be recovered when $\nm < \sps$, but
a larger number of overall measurements are needed.

\subsection{Key contributions and techniques}
  
Our main result is a characterization of the sample complexity of the support recovery problem.
We will refer to the data vectors $X_{i}$ as samples and the observations $Y_{i}$ as measurements.
We show that in the \textit{measurement-starved} regime of $(\log k)^{2} \le \nm<\sps$,
the minimum number of samples required to recover the 
support correctly with large probability is $\Theta((\sps^2/\nm^2)\log
\sps (\dmn-\sps))$.
The sample-optimal estimator we study entails
forming an estimate 
of the variance of the samples along each of the $d$ coordinates -- we
do this using a simple closed form estimate. Following this, the
coordinates corresponding to the $\sps$ largest variances give an
estimate for the support. We provide a separation condition between the values of the estimate on the true support and outside it, that guarantees exact recovery. 
Our analysis of this condition builds on concentration bounds for the powers of lengths and inner products of the columns of the measurement matrices.

In Figure
\ref{fig:ptx}, we plot the probability of exact support recovery for
our proposed closed-form estimator as a function of the normalized
number of samples. These curves confirm our theoretical results on the scaling of $\ns$, and a
clear phase transition can be seen. We discuss these results  
further in Section \ref{sec:simulations}. 
\begin{figure}[t]
	\centering
	\begin{subfigure}{0.45\linewidth} \centering
		\includegraphics[scale=0.55]{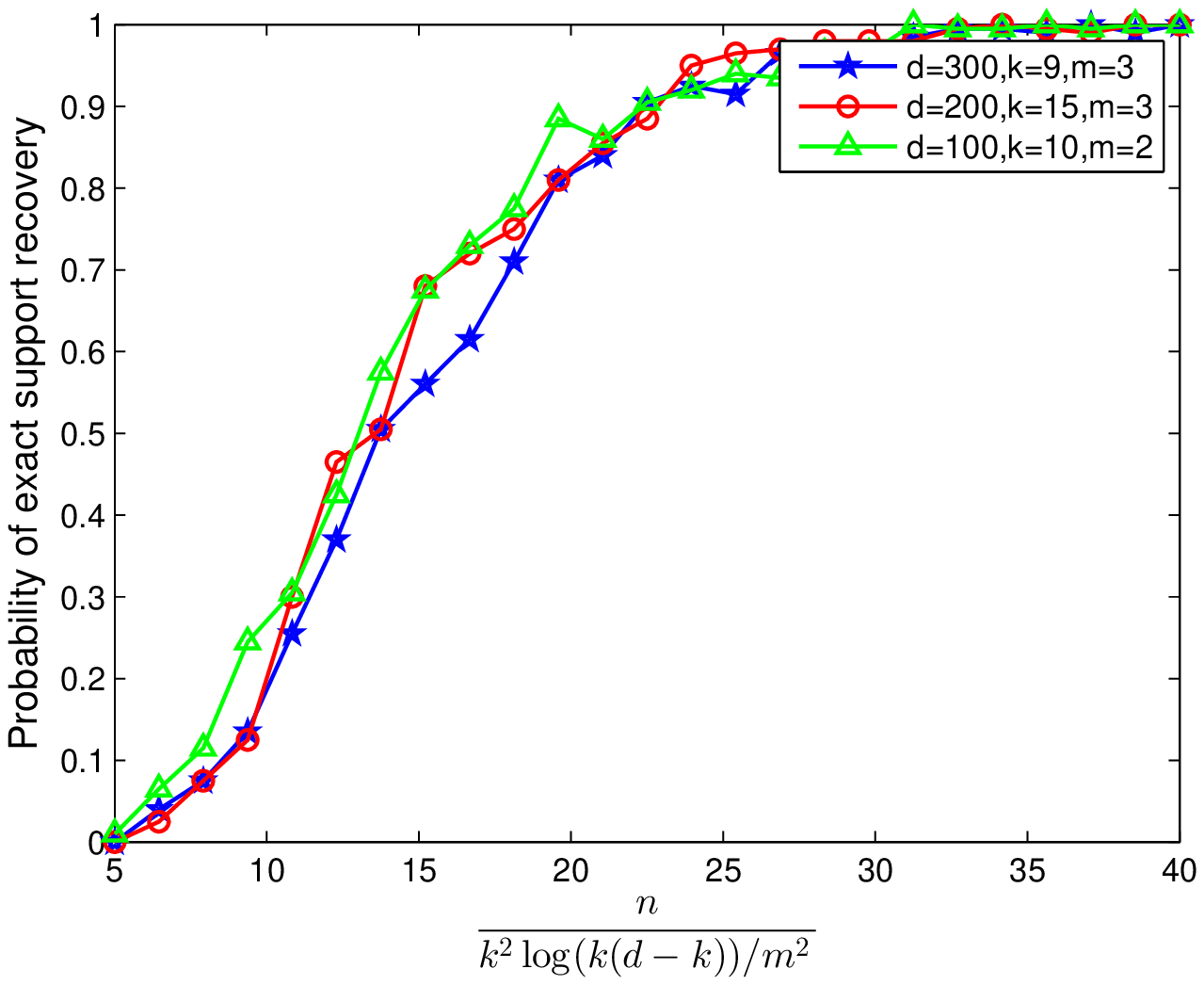}
		\caption{Gaussian prior}\label{fig:ptx_a}
	\end{subfigure}
	\begin{subfigure}{0.45\linewidth} \centering
		\includegraphics[scale=0.55]{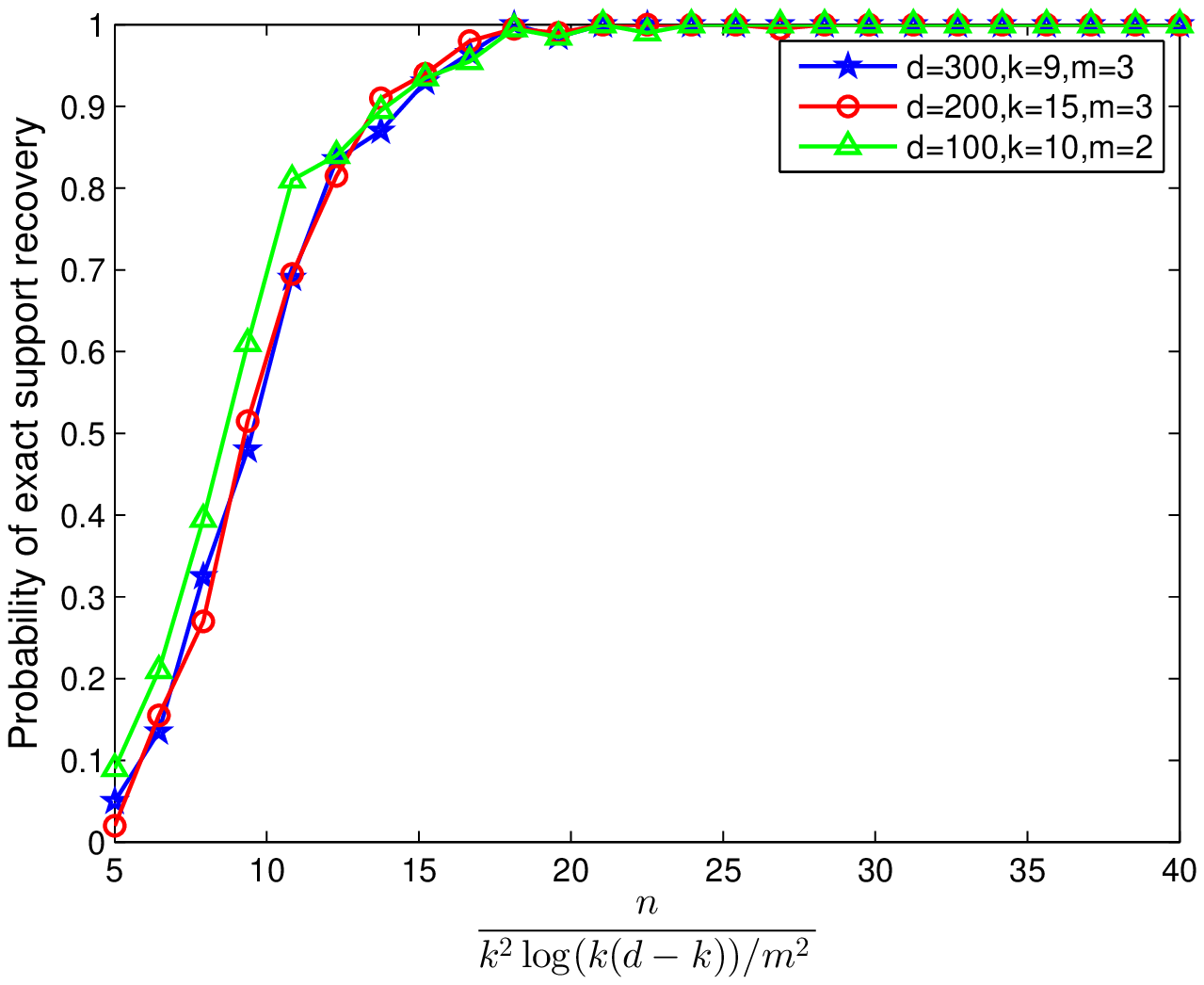}
		\caption{Rademacher prior}\label{fig:ptx_b}
	\end{subfigure}
	\caption{Phase transition of the closed-form estimator.} \label{fig:ptx}
\end{figure}


Our information-theoretic lower bound
is obtained by applying Fano's method to a difficult case with 
supports of size $k$ differing in one entry. The main challenge in the
proof is to characterize the reduction in the distances between the
distributions corresponding to different supports due to compressed linear
measurements. We capture this by a quantity related to the 
spectrum of the Gram matrix of a random Gaussian measurement matrix.

\subsection{Prior work}\label{sec:prior work}
Information-theoretically optimal support recovery in the single sample setting is well-understood; \cite{Wainwright_TIT_2009}, \cite{Fletcher_TIT_2009} and \cite{Aeron_TIT_2010} were some of the first works to look at this problem. In particular, \cite{Wainwright_TIT_2009} shows that for a deterministic input vector, $\nm=\Theta(\sps\log(\dmn-\sps))$ measurements are necessary and sufficient to exactly recover the support using a Gaussian measurement matrix,  establishing that support recovery is impossible in the $\nm<\sps$ regime using a single sample. Following these works, several papers proposed algorithms for the  multiple sample setting, that include convex programming methods \cite{Malioutov_TSP_2005}, \cite{Tropp_SP_cvx_2006}, \cite{Eldar_TIT_2010}, thresholding-based methods \cite{Foucart_SAMPTA_2011}, \cite{Gribonval_JFA_2008}, Bayesian methods \cite{Wipf_TSP_2007} and greedy methods \cite{Tropp_ICASSP_2005}, \cite{Tropp_SP_gr_2006}.  
 	Recovery of the support
 	is important in several practical applications including spectrum
        sensing \cite{Sharma_TVT_2014} and group testing
        \cite{Wang_ITA_2015}. 
 Moreover, in settings where there are
          \emph{multiple} unknown supports, the common support
          recovery algorithm can be used to estimate the union of the
          supports. This can be useful in problems of linear
          regression where there are multiple unknown subsets of
          correlated variables \cite{Obozinski_annals_stats_2011}.
         When $\nm\ge \sps$,
 	support recovery implies recovery of the data vectors
        also. Indeed, given the support, 
 	one can estimate the data vectors by solving a least squares
        problem 
 	restricted to the support.  In this work, we show that when
        $\nm <\sps$, support 
 	recovery is still possible. However, recovery of the data
        vectors is no longer possible, since there are infinitely many
        solutions even after restricting to the support.

A setup similar to ours was studied in~\cite{Park_TSP_2017}, but the
results are not tight in the $\nm
<\sps$ regime. In particular,~\cite{Park_TSP_2017} showed a lower bound
on sample complexity of support recovery of roughly $(\sps/\nm)$, much
weaker than our $(\sps/\nm)^2$ lower bound.
 Another related line of works \cite{Tang_TIT_2010}, \cite{Jin_TIT_2013} studies this problem considering the \emph{same} measurement matrix for all samples, under the assumption that the data vectors are deterministic.  However, none of these works characterize the tradeoff between $\nm$ and $\ns$ when~$\nm<\sps$.

Initial works considering the $m<k$ regime were \cite{Pal_TSP_2015} and \cite{Balkan_SPL_2014}, followed by \cite{Khanna_SPAWC_2017} and \cite{Ramesh_ICASSP_2018}, where it was empirically demonstrated that when multiple samples are available, it is possible to operate in the $\nm<\sps$ regime.
We note that the estimator we consider is similar to the one in \cite{Koochakzadeh_SPIE_2018}. 
 However, the analysis in \cite{Koochakzadeh_SPIE_2018}  is conditioned on the measurement matrix ensemble, and the error probability is expressed in terms of quantities dependent on the measurement matrix. As such, the final dependence of $n$ on $m,k$ and $d$ cannot be inferred from this result. In this work, we overcome these shortcomings by showing
  that the estimator successfully recovers the support for a large class of subgaussian measurement matrix ensembles, and we explicitly characterize the dependence of $n$ on $m,k$ and $d$. We also provide matching lower bounds, which shows the optimality of the estimator.
Two other related works that consider the $m<k$ setting are \cite{Pal_TSP_2015} and \cite{Koochakzadeh_TSP_2018}.
However, the precise characterization of sample complexity is not addressed in any of these works.

Our formulation of support recovery in a Bayesian setting naturally relates to some of the works on covariance estimation. A recent work which looks at the problem of covariance estimation from low-dimensional projections of the data is \cite{Azizyan_TIT_2018}. 
As we will see in the next section, the support recovery problem amounts to estimating a diagonal and sparse covariance matrix, and the general result in \cite[Corollary 3]{Azizyan_TIT_2018} for this specific case is loose and does not give the correct scaling for the sample complexity. Two other works that study covariance estimation from projected samples in the $\nm = 1$ case are \cite{Cai_annals_stats_2015} and \cite{Goldsmith_TIT_2015}. 
However, these results also do not give the correct scaling on the number of samples, when applied to the diagonal sparse case.
Further, since $\nm$ is set to one, the tradeoff between $\nm$ and $\ns$ is not clear.

Our setting is also related
to the recently considered inference under local information
constraints setting of~\cite{Acharya_corr_2018}. We impose information constraints
on each sample by allowing only $\nm$ linear measurements per sample. Roughly, our results say that when local information constraints are placed (namely, $\nm<\sps$), support recovery requires much more than $\sps$ overall measurements.

\subsection{Organization}
In the next section, we formally state the problem and the assumptions
we make in our generative model setting. We then state our main
result, which characterizes the sample complexity of support
recovery. In Section \ref{sec:upper_bound}, we propose and analyze a
closed-form estimator and show it to be sample-optimal. We present
simulation results, followed by the proof of the lower bound in
Section \ref{sec:lower_bound}. We extend the results to the nonbinary variance case in Section~\ref{sec:nonbinary}, and  conclude with discussion and future
work in Section \ref{sec:discussion}. 
Appendicex \ref{app_pf_sep} contains the proof of the main technical result needed
for analyzing our scheme, and builds on results from Appendix \ref{app_lemmas}. Appendix \ref{app_smin} contains a bound for
the fourth moment of the minimum eigenvalue of a Wishart matrix, which
is required for the  lower bound. Some of the background results on
concentration bounds
needed for our proofs are given in Appendix \ref{app_prelim}.

\textit{Notation and Preliminaries.} For a matrix $A_{j}$, $A_{ji}$ denotes its $i$th column, $A_{j}(u,v)$ denotes its $(u,v)$th entry and $(A_{j})_{S}$ denotes the submatrix formed by columns indexed by $S$. 
Also, for a vector $X_{j}$, $X_{ji}$ denotes the $i$th component of $X_{j}$. For symmetric matrices $A$ and $B$, the notation $A\succcurlyeq B$ denotes
that the matrix $A-B$ is positive semidefinite. 
For a vector $a\in\mR^{\dmn}$, $\texttt{diag}(a)$ denotes a diagonal matrix with the entries of $a$ on the diagonal.
A random variable $X$ is subgaussian with variance
parameter $\sigma^2$, denoted 	$X\sim\sbg(\sigma^2)$, if  
\begin{equation}\label{subg_def}
\log \bEE{e^{\theta (X-\bEE{X})}}\leq \theta^2\sigma^2/2,
\end{equation}
for all $\theta\in\mR$.
Similarly, a random variable $X$ is subexponential with parameters $v^{2}$ and $b>0$, denoted $X\sim \sbx(v^{2},b)$, if
\begin{equation}\label{subexp_def}
\log \bEE{e^{\theta(X-\bEE X)}}\le \theta^{2}v^{2}/2,
\end{equation}  
for all $|\theta|<1/b$.
We use the shorthand $X_{1}^{n}$ to denote independent and identically distributed random variables $X_{1},\ldots,X_{n}$.  
When taking expectation of a function of several random variables $Z=f(X_{1},\ldots,X_{n})$, we use $\se_{X_{1}}[Z]$ to denote that the expectation is with respect to the distribution of~$X_{1}$.


\section{Problem formulation and main result}\label{sec:formulation}

We start by considering a Bayesian formulation for support
recovery, where the input comprises $\ns$ independent samples $\bx_1,
\dots, \bx_\ns$ in $\mR^\dmn$, with each $\bx_{i}$ having a zero-mean
Gaussian distribution. We denote the covariance of $\bx_{i}$ by
$K_{\lambda}\ed \mathtt{diag}(\lambda_1, \lambda_2, \ldots,
\lambda_\dmn)$, where the $\dmn$-dimensional vector $\lambda$ has
entries $\lambda_1, \lambda_2, \ldots, \lambda_\dmn$ such that
$\lambda \in \cS_{\sps,\dmn}\ed \left\{u\in \{0,1\}^\dmn\,:\,
\normzero{u}= \sps\right\}$.  That is, the (random) data
vectors have a common support $S=\mathtt{supp}(\lambda)$ of size $\sps$, almost surely.

Each $\bx_i$ is passed through a random $\nm\times \dmn$ measurement
matrix $\bphi_i$, $1\leq i \leq \ns$, with independent, zero-mean
Gaussian entries with variance $1/\nm$, and the observations $\by_i = \bphi_i\bx_i \in\mR^\nm$ 
are made available to a center. Using the measurements
$\by_1,\ldots,\by_\ns$, the center seeks to determine the common
support  $S$ of $\bx_1,\ldots, \bx_\ns$. 
To that end, the center uses
an estimate $\hat{S}: \mR^{\nm\times \ns} \to {[\dmn]\choose \sps}$,
where ${[\dmn]\choose k}$ denotes the set of all subsets of $[\dmn]$ of cardinality $\sps$. 
We seek estimators that can recover the support of $\lambda$ accurately
with probability of error no more than $\delta \ed 1/3$,\hspace{-.2em}\footnote{Note that the value $\delta=1/3$ is chosen here for convenience and
	can be replaced with any acceptable value below~$1/2$. However, our results may not be tight in their dependence on~$\delta$.} namely
\begin{align}
\bPr{\hat{S}(Y^\ns)\neq \mathtt{supp}(\lambda)}\leq
\delta, \quad \forall\, \lambda \in \cS_{\sps, \dmn}.
\label{e:error_criterion}
\end{align}
In the compressed sensing literature, this is usually referred to as a non-uniform recovery guarantee.

We are interested in sample-efficient estimators. The
next definition introduces the fundamental quantity of interest to
us.
\begin{definition}[Sample complexity of support recovery]
	For $\nm, \sps, \dmn\in \mN$, the sample complexity of support recovery 
	$\ns^*(\nm, \sps, \dmn)$ is defined as the minimum number of
	samples $\ns$ for which we can find an estimator $\hat{S}$
	satisfying \eqref{e:error_criterion}.
\end{definition}

Our main result is the following.
\begin{theorem}[Characterization of sample complexity]\label{t:main}
	For {$(\log k)^{2}\leq \nm<\sps/2$} and $1\le\sps\le\dmn-1$, the sample complexity of support
	recovery in the setting above is given by
	\begin{equation*}
	\ns^*(\nm, \sps, \dmn)=\Theta\left(\frac{\sps^2}{\nm^2}\cdot \log\sps(\dmn-\sps)\right). \label{eq:ULbound}
	\end{equation*}
\end{theorem}
\begin{remark}
	{Our formulation  assumes that the support size $\sps$
	is known. That said, our proposed estimator extends easily to the setting
	where we only have an upper bound of $\sps$ on the support size, and we
	seek to output a set of indices containing the support.}
\end{remark}
\begin{remark}
	{We expect the scaling in Theorem \ref{t:main} to hold good even when $m<(\log k)^{2}$. In fact, our lower bound result continues to hold for $m=1$. The current upper bound proof, however, requires $\nm\ge(\log k)^{2}$. We discuss this further in Appendix \ref{app_pf_sep}}. 
\end{remark}

We provide the optimal estimator and prove the upper bound in 
Section~\ref{sec:upper_bound} and the information-theoretic lower bound in Section~\ref{sec:lower_bound}.
Our proof yields a lower bound when $\nm<\alpha \sps$ for any $\alpha<1$.

Our proposed estimator and its analysis applies to a much
broader setting involving subgaussian priors (see \eqref{subg_def} for definition).
For $\bx_1^\ns$, we can use any prior with
subgaussian distributed entries, i.e., the entries of $\bx_i$ are independent and zero-mean with 
$\bEE{\bx_{i,j}^2}=\lambda_j$ for $\lambda\in  \cS_{\sps, \dmn}$ and  
$\bx_{i,j}\sim\sbg(\lambda_j^\prime)$, where $\lambda_j^\prime$ is the
variance parameter for the subgaussian 
random variable $\bx_{i,j}$. Our analysis will go through as long as
the variance and variance
parameters differ only up to a constant factor. 

Also, the measurement matrices $\bphi_i$ can be chosen to have
independent, zero-mean subgaussian distributed entries in place of
Gaussian. However, as above, we assume that the variance and variance
parameter of each entry are the same up to a multiplicative constant
factor. Further, we assume that the fourth moment of the entries of $\bphi_{i}$ is of the order of the square of the
variance. Two important ensembles that satisfy these properties are the Gaussian ensemble and the Rademacher ensemble.

For clarity, we restate our assumptions below. These assumptions are required for the analysis of our estimator; the lower bound proof is done under the more restrictive setting of Gaussian measurement matrix ensemble (which implies a lower bound for the subgaussian ensemble also). 
\begin{assumption}\label{assump_1}
		The entries of $\bx_{i}$, $i\in[\ns]$, are independent and zero-mean with 
		$\bEE{\bx_{i,j}^2}=\lambda_j$ for $\lambda\in  \cS_{\sps, \dmn}$ and  
		$\bx_{i,j}\sim\sbg(c \lambda_j)$, where $c$ is an absolute constant;
\end{assumption}
\begin{assumption}\label{assump_2}
	The entries of $\bphi_{i}$, $i\in[\ns]$, are independent and zero-mean with $\bEE{\bphi_{i}(u,v)^{2}}=1/\nm$,  $\bphi_{i}(u,v)\sim\sbg(c^\prime/\nm)$,  
	and $\bEE{\bphi_{i}(u,v)^{4}}=c^{\prime\prime}/\nm^{2}$,  
	where $c^\prime$ and $c^{\prime\prime}$ are absolute constants.
\end{assumption}

	Our results also extend to the case when the data vectors are not necessarily sparse in the standard basis for $\mR^{\dmn}$, i.e., the data vectors can be expressed as $\bx_{i}=B Z_{i}$, $i\in[\ns]$, where $B$ is any orthonormal basis for $\mR^{\dmn}$ and $Z_{i}$s have a common support of size $\sps$. Under the same generative model as before, but for $Z_{i}$ this time, Theorem \ref{t:main} continues to hold. This is because when $\bphi_{i}$ is Gaussian, the \emph{effective} measurement matrix $\bphi_{i}B$ also satisfies the properties we mentioned above, namely, it has independent mean zero Gaussian entries with variance $1/\nm$ and fourth moment $3/\nm^{2}$.

{We have restricted $\lambda$ to binary vectors for the ease of presentation. Later, in Section~\ref{sec:nonbinary}, we will show that our sample complexity results extend almost verbatim to a more general setting with the nonzero coordinates of $\lambda$ taking values between $\lambda_{\min}$ and $\lambda_{\max}$. The only change, in effect, is an additional factor $(\lambda_{\max}/\lambda_{\min})^2$ in the sample complexity of support recovery.}

Finally, we  
can allow noisy measurements $\by_i = \bphi_i\bx_i+\bw_i \in \mR^\nm$
where the noise $\bw_i$ has 
independent, zero-mean subgaussian entries independent of $\bx_i$ and
$\bphi_i$, with variance parameter $\sigma^2$.  

We present the upper bound for this more general setting, along with
our proposed estimator, in the next section. 
\section{The estimator and its analysis}\label{sec:upper_bound}
We will work with the more general  setting described
in the previous section, that is, with subgaussian random variables satisfying assumptions \ref{assump_1} and \ref{assump_2}. In fact, for simplicity, we assume that $\bx_{i,j}$
and $\bw_i$ are subgaussian with variance parameter equal to their
respective variances, {a property known as \textit{strict subgaussianity}}. Also, for the measurement matrix, we work with
the same parameters as those for the 
Gaussian ensemble and set 
\begin{equation*}
\bEE{\bphi_{i}(u,v)^2}=\frac 1\nm, \qquad \bEE{\bphi_{i}(u,v)^4}=\frac
{3}{\nm^2},
\end{equation*}
and assume  that $\bphi_{i}(u,v)$ is
subgaussian with variance parameter $1/m$. These assumptions of
equality can be relaxed to order equality up to multiplicative constants.

\subsection{The estimator}
We now present a closed-form estimator for
$\lambda$. To build heuristics, consider the trivial case where we can
directly access samples
$\{\bx_{i}\}_{i=1}^{n}$. Then, a
natural estimate for $\lambda_{i}$
is the sample variance. But in our setting, we only have access to the
measurements $\{\by_{i}\}_{i=1}^{n}$. We compute the
vector $\bphi_{i}^{\top}\by_{i}$ and treat it as a ``proxy" for
$\bx_{i}$. When $\bphi_{i}^{\top}\bphi_{i}=I$ and the measurements are noiseless, this proxy will indeed
coincide with $\bx_i$. We compute the sample variances using these new
proxy samples and use it to find an estimate for the support
	of $\lambda$.

Formally, we consider the estimate $A$ for the covariance matrix of
$\bx_{i}$s given by
\begin{equation*}
A=\frac{1}{\ns}\sum_{j=1}^{\ns}\bphi_{j}^{\top}\by_{j}\by_{j}^{\top}\bphi_{j}.
\end{equation*}
Note that $A$ is positive semidefinite. We form an intermediate
estimate $\tilde{\lambda}$ for the variance vector $\lambda$ using the
diagonal entries of $A$ as follows:
\begin{align}
\tilde{\lambda}_{i}\ed
A_{ii}=\frac{1}{\ns}\sum_{j=1}^{\ns}(\bphi_{j}^{\top}\by_{j}\by_{j}^{\top}\bphi_{j})_{ii} 
=\frac{1}{\ns}\sum_{j=1}^{\ns}(\bphi_{ji}^{\top}\by_{j})^2,   
\numberthis \label{lambda_tilde}
\end{align}
where $\bphi_{ji}$ denotes the $i$th column of $\bphi_{j}$.  Since we
are only interested in estimating the support, we simply declare indices corresponding to the largest 
$\sps$ entries of $\tilde{\lambda}$ as the support, namely, we sort $\tilde{\lambda}$ to get 
$\tilde{\lambda}_{(1)}\geq \tilde{\lambda}_{(2)} \geq \dots \geq \tilde{\lambda}_{(\dmn)}$ and output
\begin{align}
\tilde{S}= \{(1) , ..., (\sps)\},
\label{supp-est-sort}
\end{align}
where $(i)$ denotes the index of the $i$th largest entry in $\tilde{\lambda}$. 
Note that evaluating $\tilde{\lambda}_i$ requires $O(\ns \nm)$ steps,
whereby the overall computational complexity of (naively) evaluating
our proposed estimator is $O(\dmn \ns \nm)$. 
As noted in Section \ref{sec:prior work}, this estimator is similar to the one considered in \cite{Koochakzadeh_SPIE_2018}. However, unlike \cite{Koochakzadeh_SPIE_2018}, our upper bound analysis provides explicit dependence on all the parameters involved,
and is complemented by a matching lower bound result.


{Before we move to detailed analysis in the next section, we do a quick sanity test for our estimator 
and evaluate its ``expected behavior''.}
An easy calculation shows that $\tilde{\lambda}_{i}$ is an estimate of
$\lambda_{i}$ with a bias depending only on $k,m,$ and $\sigma^2$.
In particular, we have the following result.
\begin{lemma}\label{lem_bias}
Let the estimator $\tilde{\lambda}$ be as defined in \eqref{lambda_tilde}. Then, under Assumptions \ref{assump_1} and \ref{assump_2} with $c=c^{\prime}=c^{\prime\prime}=1$, we have that	
\begin{align*}
\bEE{\tilde{\lambda}_{i}}
=\frac{\nm+1}{\nm}\lambda_{i}+\frac{\sps}{\nm}+\sigma^2,\quad i\in[d],
\end{align*}
where the expectation is with respect to the joint distribution of
$\{\bx_{1}^{\ns},\bphi_{1}^{\ns},\bw_{1}^{\ns}\}$.
\end{lemma}

\begin{proof}	
We can rewrite the estimator in \eqref{lambda_tilde} as
	\begin{align*}
\tilde{\lambda}_{i}
	=\frac{1}{\ns}\sum_{j=1}^{\ns}\bigg(\sum_{l\in S}\bx_{jl}(\bphi_{ji}^{\top}\bphi_{jl})+\bphi_{ji}^{\top}\bw_{j}\bigg)^2,\quad i\in[\dmn].
	\end{align*}
	
	Taking expectation, we note that,
	\begin{align*}
	\bEE{\tilde{\lambda}_{i}}
		&=\se_{\bphi_{1}^{\ns}}\bigg[\frac{1}{\ns}\sum_{j=1}^{\ns}\bigg(\se_{\bx_{1}^{\ns}}\bigg[\sum_{l\in S}\bx_{jl}^{2}(\bphi_{ji}^{\top}\bphi_{jl})^{2}\bigg] 
		+\se_{\bw_{1}^{\ns}}\bigg[(\bphi_{ji}^{\top}\bw_{j})^{2}\bigg]\bigg)\bigg\rvert\bphi_{1}^{\ns}\bigg]\\
	&=\se_{\bphi_{1}^{\ns}}\bigg[\frac{1}{\ns}\sum_{j=1}^{\ns}\bigg(\sum_{l\in S}(\bphi_{ji}^{\top}\bphi_{jl})^2+\sigma^2\|\bphi_{ji}\|_{2}^{2}\bigg)\bigg],
	\end{align*}
	where the second step uses the fact that $\bx_{j}$ has zero mean entries.

	Using our assumption that the columns of $\bphi_{j}$ have independent mean-zero entries with variance $1/\nm$ and fourth moment $3/\nm^{2}$, it follows from Lemma \ref{lem_moments} in the Appendix that
	\begin{equation*}
	\bEE{\tilde{\lambda}_{i}}=
	\frac{\nm+1}{\nm}\lambda_{i}+\frac{\sps}{\nm}+\sigma^2,\quad i\in[\dmn],
	\end{equation*}
	which establishes the lemma.
\end{proof}

We work with this biased $\tilde{\lambda}$ and analyze its performance in the next section. Since the bias is the same across all coordinates, it does not affect sorting/thresholding based procedures. 
{The key observation here is that the expected values of
the coordinates of $\tilde{\lambda}$ 
 in the support of $\lambda$ exceeds those outside the support, making it an appropriate statistic for support recovery.}

\subsection{And its analysis}\label{sec:p_err}
A high level overview of our analysis is as follows. We first note
that, conditioned on the 
measurement matrices, the entries of $\tilde{\lambda}$ are sums of independent
subexponential random variables (defined in \eqref{subexp_def}). If we can ensure that there is
sufficient separation between the typical values of $\tilde{\lambda}_i$ in the
$i\in S$ and $i^\prime\in S^{c}$ cases, then we can 
distinguish between the two cases. We show that such a separation
holds with high probability for our subgaussian measurement matrix ensemble.

We now present the performance of our estimator.
\begin{theorem}\label{t:upper_bound}
	Let $\tilde{S}$ be the estimator described in
	\eqref{supp-est-sort} and assume that~{$(\log k)^{2}\leq \nm$}. 
 Then, under Assumptions \ref{assump_1} and \ref{assump_2}, $\tilde{S}$
	equals the true support with probability at least $1-\delta$ provided
	\begin{align}\label{thm_ub}
	\ns\ge
	c\bigg(\frac{\sps}{\nm}+1+\sigma^{2}\bigg)^2\log\frac{\sps(\dmn-\sps)}{\delta},
	\end{align}
	for an absolute constant $c$. 
\end{theorem}
\begin{remark}
We note that the result above applies for all $\sps$ and all {$\nm>(\log k)^2$}, and not only to our regime of interest $\nm<\sps$.  
When $\sigma^2=0$ and $k/m>1$, we
obtain the upper bound claimed in Theorem~\ref{t:main}. 	
\end{remark}
\begin{remark}
	Our lower bound on sample complexity implied by \eqref{eq:ULbound} in Theorem~\ref{t:main} is for the noiseless case and therefore does not capture the dependence on $\sigma^{2}$. However, simulation results (see Figure~\ref{fig:ptxnoise} in Section~\ref{sec:simulations}) suggest that the dependence in \eqref{thm_ub} is tight.
		
\end{remark}
\begin{proof}
	While computationally tractable, analyzing our proposed estimator directly may not be easy. Instead, we 
	analyze an alternative {thresholding-based} estimator given by
	\begin{equation}\label{cf_est}
	\hat{\lambda}_{i}=\mathbbm{1}_{\{\tilde{\lambda}_{i}\ge\tau\}}.
	\end{equation}
	We note that if $\lambda=\hat{\lambda}$, the
	largest $\sps$ entries of $\tilde{\lambda}$ must coincide with the
	support of $\lambda$. Therefore, 
	\begin{align}
	\bPr{\tilde{S}\neq \mathtt{supp}(\lambda)}\leq
	\bPr{\hat{S}\neq\mathtt{supp}(\lambda)},
	\label{e:error_relation}
	\end{align}
	where $\hat{S}$ is the support of $\hat{\lambda}$. Using this observation, it suffices to analyze the estimator
	$\hat{\lambda}$ in~\eqref{cf_est}, which will be our focus below.
	
The proof of Theorem~\ref{t:upper_bound} entails a careful analysis of 
tails of $\tilde{\lambda}_i$ and uses standard subgaussian and
subexponential concentration bounds. {To bound the error term in \eqref{e:error_relation}, we rely on the measurement matrix ensemble satisfying a certain separation condition; we denote this event by $E$ and describe it in detail shortly.
Denoting by $S$ the support of $\lambda$, we note that  $\bPr{\hat{S}\ne S}$ can be bounded as
\begin{align*}
\bPr{\hat{S}\ne S}
&\le \bPr{\hat{S}\ne S|E}+\bPr{E^{c}}\\
&\le \sum_{i\in S}
\bPr{\tilde{\lambda}_{i}<\tau|E}+\sum_{i^\prime\in
	S^c}\bPr{\tilde{\lambda}_{i^\prime}\ge\tau|E}+\bPr{E^{c}}.\numberthis \label{p_err_est}
\end{align*}
{We show that the first two terms in the equation above, 
involving probabilities conditioned on the event $E$,} 
can be made small. Also, for the subgaussian measurement ensemble, $E$ occurs with large probability, which in turn implies that the overall error can be made small.}

Our approach
involves deriving tail bounds for $\tilde{\lambda}_i$
conditioned on the measurement matrices, and then
choosing a threshold $\tau$ to obtain the desired bound
for~\eqref{p_err_est}; we derive lower tail bounds for $i\in S$ and
upper tail bounds {for $i^\prime\in S^c$}. {The event $E$ mentioned above corresponds to the measurement ensemble being such that we can find a threshold $\tau$ that allows us to separate these bounds.}

Specifically, note that 
\begin{align*}
\tilde{\lambda}_{i}
&=\frac{1}{\ns}\sum_{j=1}^{\ns}\bigg(\sum_{l\in
	S}\bx_{jl}(\bphi_{ji}^{\top}\bphi_{jl})+\bphi_{ji}^{\top}\bw_{j}\bigg)^2,
\end{align*}
where we used $\by_{j}=\bphi_{j}\bx_{j}+\bw_{j}$. Conditioned on $\bphi_{1}^{\ns}$,
$\tilde{\lambda}_{i}$ is a sum of independent subexponential random
variables. 
Using properties of subexponential random variables and the
connection between subgaussian and subexponential random variables (described in Lemmas \ref{lem_sbx_tail} and  \ref{lem_sbx_sbg} in the appendix), we
get that conditioned on the measurement matrices $\bphi_{1}^{\ns}$, 
	the random variable $\tilde{\lambda}_{i}$ is
\begin{align*}
\sbx\bigg(\frac{c_1}{\ns^2}\sum_{j=1}^{\ns}\alpha_{ji}^4,\frac{c_2}{\ns}\displaystyle\max_{j\in[\ns]}\alpha_{ji}^2\bigg),
\end{align*}
where $c_1$ and $c_2$ are absolute constants and
\begin{align*}
\alpha_{ji}^{2}=\begin{cases}
\|\bphi_{ji}\|_{2}^{4}+\!\!\!\displaystyle\sum_{l\in
	S\backslash\{i\}}\!\!\!(\bphi_{jl}^{\top}\bphi_{ji})^{2}+\sigma^2\|\bphi_{ji}\|_{2}^{2},\quad
i\in S,\\ \displaystyle\sum_{l\in
	S}(\bphi_{jl}^{\top}\bphi_{ji})^{2}+\sigma^2\|\bphi_{ji}\|_{2}^{2},\quad
\text{otherwise}.
\end{cases}
\end{align*}
Using standard tail bounds for
subexponential random variables given in Lemma \ref{lem_sbx_tail}
 and denoting $\mu_{i}\ed\bEE{\tilde{\lambda}_{i}|\bphi_{1}^{\ns}}=\frac{1}{\ns}\sum_{j=1}^{\ns}\alpha_{ji}^{2},
	i\in[\dmn]$, we have for
	$i\in S$, 
	\begin{equation*}
	\bPr{\tilde{\lambda}_{i}<
		\tau\rvert\bphi_{1}^{\ns}}
	\le\exp\bigg(-\min\bigg\{\frac{\ns^2
		(\mu_{i}-\tau)^2}{c_1\sum_{j=1}^{\ns}\alpha_{ji}^4},\frac{\ns
		(\mu_{i}-\tau)}{c_2\displaystyle\max_{j\in
			[\ns]}\alpha_{ji}^2}\bigg\}\bigg),
	\end{equation*}
	and {for $i^\prime\in S^{c}$},
	\begin{equation*}
	\bPr{\tilde{\lambda}_{i^\prime}\geq
		\tau\rvert\bphi_{1}^{\ns}} \le\exp\bigg(-\min\bigg\{\frac{\ns^2
		(\tau-\mu_{i^\prime})^2}{c_1\sum_{j=1}^{\ns}\alpha_{ji^\prime}^4},\frac{\ns
		(\tau-\mu_{i^\prime})}{c_2\displaystyle\max_{j\in
			[\ns]}\alpha_{ji^\prime}^2}\bigg\}\bigg).
	\end{equation*}

We can upper bound the sum of the first two terms in \eqref{p_err_est} by $\delta/2$
by showing that with large probability $\bphi_{1}^{\ns}$
takes values for which we get each term above bounded by roughly
$\delta^{\prime}\ed\delta/(4\max\{(\dmn-\sps), \sps\})$. In particular, using a manipulation of
the expression for exponents, each of the conditional probabilities
above will be less than
$\delta^{\prime}$ if  $\tau$ satisfies the following condition for any
$i\in S$ and $i^\prime\in S^{c}$: 
\begin{align}\label{sep_0}
\mu_{i^\prime}+\nu_{i^\prime}\le\tau\le\mu_{i}-\nu_{i},\numberthis
\end{align}
where
\begin{align*}
\nu_{i}&\ed\max\bigg\{\sqrt{\frac{c_1}{\ns^2}\sum_{j=1}^{\ns}\alpha_{ji}^4\log\frac{1}{\delta^{\prime}}},\frac{c_2}{\ns}\max_{j\in[\ns]}\alpha_{ji}^2\log\frac{1}{\delta^{\prime}}\bigg\},
\end{align*} 
and a similar definition holds for $\nu_{i^{\prime}}$. Thus, the sufficient condition in \eqref{sep_0} can be rewritten as 
\begin{align}\label{sep}
\frac{1}{\ns}\sum_{j=1}^{\ns}\alpha_{ji}^2-\frac{1}{\ns}\sum_{j=1}^{\ns}\alpha_{ji^\prime}^2
\ge &
\max\bigg\{\sqrt{\frac{c_1}{\ns^2}\sum_{j=1}^{\ns}\alpha_{ji}^4\log\frac{1}{\delta^{\prime}}},\frac{c_2}{\ns}\max_{j\in[\ns]}\alpha_{ji}^2\log\frac{1}{\delta^{\prime}}\bigg\}
 \nonumber\\
&+
\max\bigg\{\sqrt{\frac{c_1}{\ns^2}\sum_{j=1}^{\ns}\alpha_{ji^{\prime}}^4\log\frac{1}{\delta^{\prime}}},\frac{c_2}{\ns}\max_{j\in[\ns]}\alpha_{ji^{\prime}}^2\log\frac{1}{\delta^{\prime}}\bigg\}.
\end{align}

{Let $E$ denote the event that for all $i\in S$ and $i^\prime \in S^c$, condition
	\eqref{sep}
	is satisfied by the measurement matrix ensemble. We will show that for $\bphi_{1}^{\ns}$ drawn from the subgaussian ensemble satisfying assumption~\ref{assump_2}, the event $E$ occurs with high probability.}
We establish this claim by showing that each term in \eqref{sep} concentrates well around its expected value and roughly $\ns\nm^2\ge c\sps^2\log(1/\delta^{\prime})$ suffices to guarantee that the separation required in \eqref{sep} holds with large probability. The following result, which we prove in Appendix \ref{app_pf_sep}, shows that \eqref{sep} holds with large probability for all pairs $(i,i^{\prime})\in S\times S^{c}$.
\begin{lemma}\label{lem:sep}
		The separation condition 
\eqref{sep}
	holds simultaneously for all pairs $(i,i^{\prime})\in S\times S^{c}$ with probability at
	least $1-\delta$, over the choice of $\bphi_{1}^{n}$, $X_{1}^{n}$ and $W_{1}^{n}$, if $n\geq c (\sps/\nm +\sigma^2)^2\log(1/\delta^{\prime})$, where $\delta^{\prime}=\delta/(4~\text{max}\{k,d-k\})$.
\end{lemma}
{Choosing the probability parameter to be $\delta/2$ in Lemma \ref{lem:sep}, we see that the third term in \eqref{p_err_est} can be at most $\delta/2$, leading to an overall error probability of at most $\delta$.
	Further, noting that $2\log(1/\delta^{\prime})\ge \log(16k(d-k)/\delta)$, we obtain the result claimed in the theorem.}
%
\end{proof}

	\begin{remark}
		As long as the noise variance is sufficiently small, i.e., $\sigma^{2}<\sps/\nm$, our estimator is sample-optimal and achieves the same scaling as the lower bound that we prove in Section \ref{sec:lower_bound}. 
		\end{remark}
	
	\begin{remark}
		The separation condition \eqref{sep} fails to hold for $\ns=1$, regardless of which measurement ensemble is used. This is to be expected when $\nm<\sps$, since from our lower bound for sample complexity, multiple samples are necessary in the $\nm<\sps$ regime. 
\end{remark}

%

\subsection{Simulation results}\label{sec:simulations}
In this section, we numerically evaluate the performance of the closed-form estimator in \eqref{supp-est-sort}. 
Our focus will be on exact support recovery and we will study the performance of our estimator over multiple trials. For our experiments, we use measurement matrices that are independent across samples and have i.i.d. $\cN(0,1/\nm)$ entries. To generate measurements, we first pick a support uniformly at random from all possible supports of size $\sps$. Next, the data vectors are generated according to one of two methods. In the first method, the nonzero entries of the data have i.i.d. $\cN(0,1)$ entries. In the second method, the nonzero entries are i.i.d Rademacher (i.e., $\{+1,-1\}$-valued with equal probability). Both these distributions are subgaussian with variance parameters that are a constant multiple of the respective variances. We generate noiseless measurements $\by_{1}^{n}$ according to the linear model described before.
For a fixed value of $\dmn$, $\sps$, $\nm$ and $\ns$, we generate multiple instances of the problem and provide it as input to the estimator. For every instance, we declare success or failure depending on whether the support is exactly recovered or not and the success rate is the fraction of instances on which the recovery is successful. For our experiments, we performed $200$ trials for every set of parameters. We can see from Figure \ref{fig:ptx} in Section \ref{sec:introduction} that the experimental results closely agree with our predictions. Also, the constant of proportionality is small, roughly between $15$ and $20$.
%

\begin{figure}
	\begin{minipage}[t]{.475\textwidth} 
		\includegraphics[width=\textwidth]{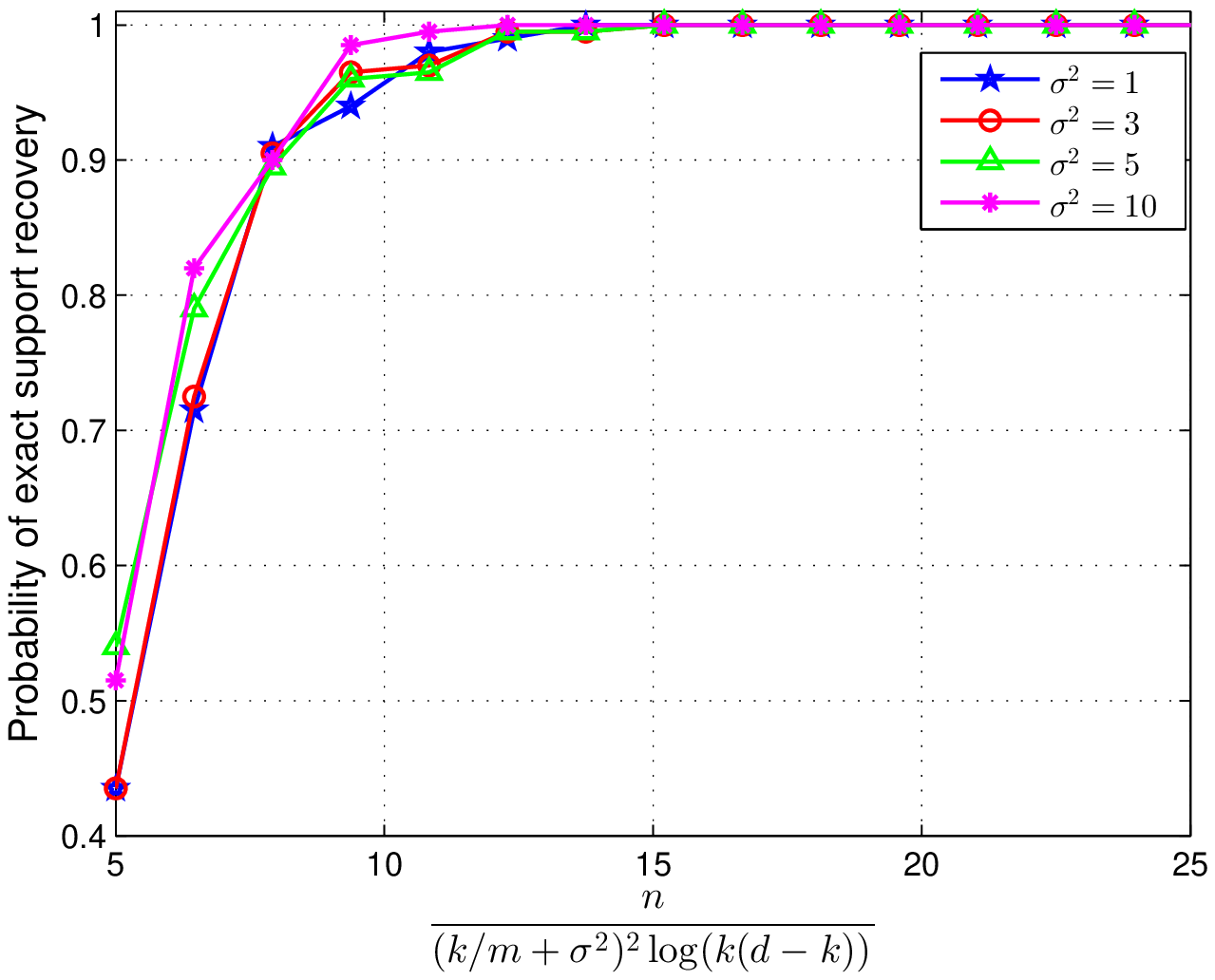}
		\caption{Performance of the closed-form estimator for different noise levels with $d=100,~m=2,~k=10$.}
		\label{fig:ptxnoise}
	\end{minipage}%
	\hfill 
	\begin{minipage}[t]{.475\textwidth}
		\includegraphics[width=\textwidth]{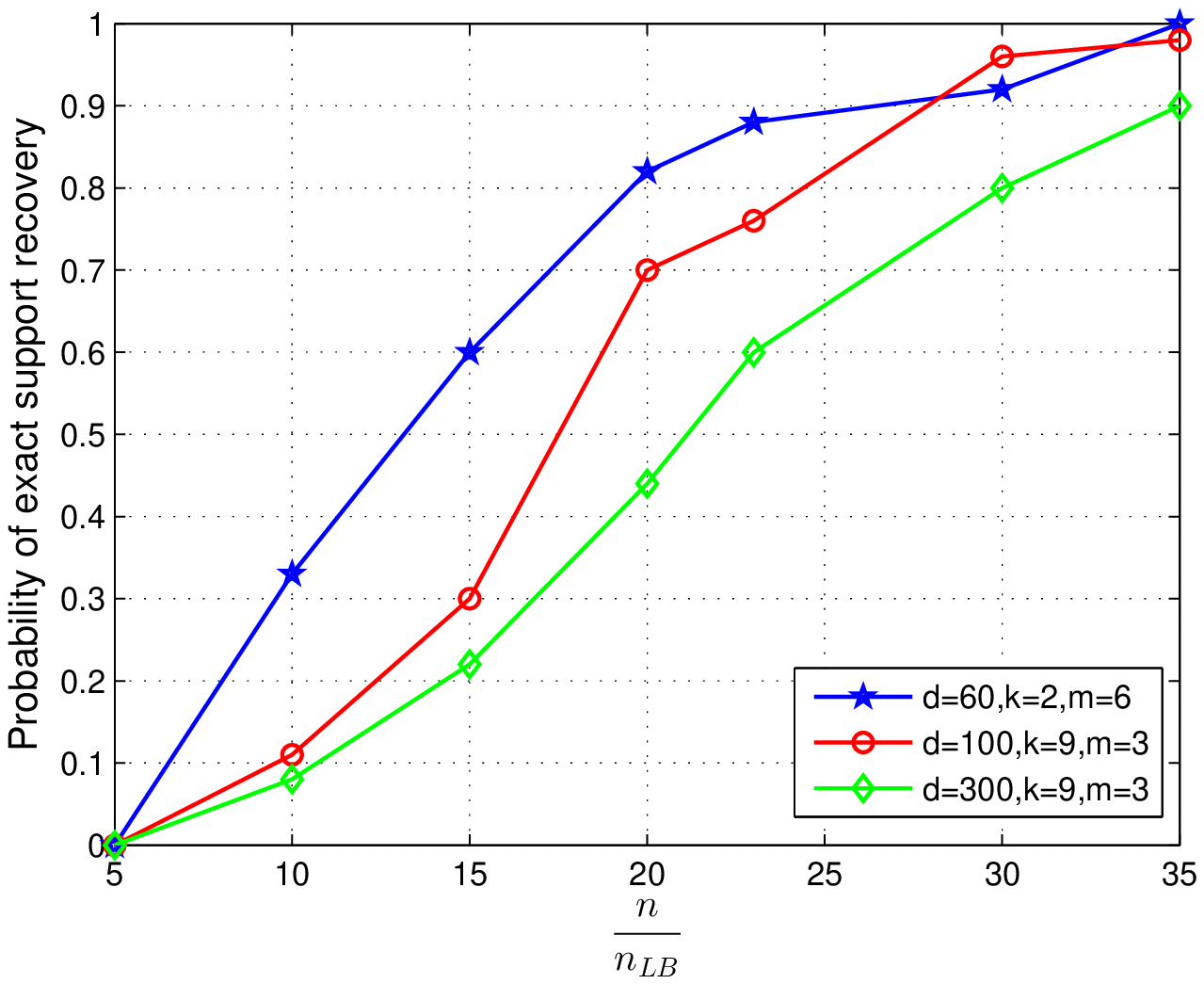}
		\caption{Performance of MSBL in the noiseless case for different parameter values.}
		\label{fig:ptxmsbl}
	\end{minipage}
\end{figure}

 
{ We also perform simulations for the case when the measurements are noisy. In particular, we consider noise vectors $\bw_{i}\stackrel{iid}\sim\cN(0,\sigma^{2}I),~i\in[n]$, and $\bx_{1}^{n}$ Gaussian distributed as described before. We plot the probability of exact support recovery against the normalized number of samples for four different noise values of the noise variance, while the other parameters are kept fixed at $d=100$, $m=2$, and $k=10$. It can be seen from Figure \ref{fig:ptxnoise} that the four curves overlap, indicating that the scaling of $n$ with respect to the noise variance is tight. Finally, Figure \ref{fig:ptxmsbl} shows the performance of MSBL \cite{Wipf_TSP_2007}, where we plot the probability of exact support recovery against the normalized number of samples (the normalization factor $n_{LB}=(k^{2}(1-m/k)^{4}/m^{2}) \log(k(d-k))$ is from the lower bound established in the next section). It can be seen that the curves do not overlap, indicating that MSBL has a different scaling of $n$ with respect to the parameters $m,k,d$ than what is predicted in the lower bound.}


\section{Lower bound}\label{sec:lower_bound}

In this section, we prove the lower bound for sample complexity {claimed in Theorem~\ref{t:main}}. 
We work with the Gaussian setting, where each
$\bphi_{i}$ has independent, zero-mean Gaussian entries with variance
$1/m$. Denote by $S_0$ the set $\{1, \ldots,
\sps\}$ and by $S_{i,j}$, $1\leq i\leq \sps< j \leq \dmn$, the set
obtained by replacing the element $i$ in $S_0$ with $j$ from
$S_0^c$. Let $U$ be distributed uniformly over the pairs $\{(i,j):
1\leq i \leq \sps, \sps+1\leq j \leq \dmn \}$. The unknown support is set to be $S_U$; the random variables
$\bx_{1}^{\ns}$ and linear measurements $Y_i=\bphi_i \bx_i$ are generated as
before. 

We consider the Bayesian hypothesis testing problem where we observe $Y^\ns$ and seek to determine $U$.
Given any support estimator $\hat{S}$, we can use it to find an estimate for the support, which in turn will give an estimate $\hat U$ for $U$. 
Clearly, $\bPr{\hat{U}\ne U }$ equals $\bPr{\hat{S}\neq S_U}$, which must be less than $1/3$ by our assumption. On the other hand, by Fano's inequality, we get
\begin{align*}
\bPr{\hat{U}\neq U}
&\ge
1-\frac{I(\by_{1}^{\ns};U)+1}{\log (\sps(\dmn-\sps))}\\  
&\geq 1- \frac{\max_{u}D(\bPP{Y^n|S_u}\| \bPP{Y^n|S_0})+1}{\log (\sps(\dmn-\sps))},
\end{align*}
where $\bPP{Y^n|S}$ denotes the distribution of the measurements when
the support of $\lambda$ is $S$ (a proof for the second inequality can be found in \cite[Theorem 21]{Csiszar_NOW_2004}). Note that $\bPP{Y^n|S}=
\prod_{i=1}^\ns\bPP{Y_i|S}$ with each $\bPP{Y_i|S}$ having the same
distribution which we denote by $\bPP{Y|S}$. Thus, 
$D(\bPP{Y^n|S_u}\| \bPP{Y^n|S_0}) =  nD(\bPP{Y|S_u}\| \bPP{Y|S_0})$. 

Next, we bound $D(\bPP{Y|S_u}\| \bPP{Y|S_0})$.
Denote by $\bphi_S$ the
$\nm\times \sps$ submatrix of $\bphi$ obtained by restricting to the
columns in $S$ and by $A_S$ the Gram matrix $\bphi_S\bphi_S^{\top}$ of
$\bphi_S$. Further, let $a_{1}\ge\ldots \ge a_{\nm}> 0$ and $b_{1}\ge\ldots \ge
b_{\nm}> 0$ be the respective eigenvalues of $A_{S_u}$ and
$A_{S_0}$. Note that $a_m>0$ and $b_m>0$ hold with probability $1$ since $\nm\le\sps$.

Denoting by $\bPP{Y|S,\bphi}$ the conditional distribution of the
measurement when the measurement matrix is fixed to $\bphi$,  we get
\begin{align*}
D\big(\bPP{Y|S_u,\bphi}\|\bPP{Y|S_0,\bphi}\big)
&=\frac{1}{2}\bigg(\log\frac{|A_{S_0}|}{|A_{S_u}|}+\text{Tr}(A_{S_0}^{-1}A_{S_u})-m\bigg)
\\ &\leq\frac{1}{2}\sum_{i=1}^{\nm}\bigg(\log\frac{b_{i}}{a_{i}}-\bigg(1-\frac{a_{i}}{b_{i}}\bigg)\bigg)\\
 &\le
\frac{1}{2}\sum_{i=1}^{\nm}\frac{(a_{i}-b_{i})^{2}}{a_{i}b_{i}},
\end{align*}
where in the first inequality holds by Lemma \ref{lem_traceAB} and the second inequality holds since {$\log x +  (1-x)/x
\leq (x-1)^2/x$} for all $x> 0$.
Using convexity of the KL divergence, we get
\begin{align*}
D\big(\bPP{Y|S_u}\|\bPP{Y|S_0}\big)
&\leq \frac 12\bEE{\sum_{i=1}^{\nm}\frac{(a_{i}-b_{i})^{2}}{a_{i}b_{i}}}\\
&\leq \frac 12\bEE{\sum_{i=1}^{\nm}\frac{(a_{i}-b_{i})^{2}}{a_{\nm}b_{\nm}}}.
\end{align*} 
Note that the expression on the right does not depend on our choice of
$u$; we fix $u=(1, \sps+1)$.
With an abuse of notation, we denote by $\bphi_j$ the $j$th column of a 
random matrix $\bphi$ with independent $\mathcal{N}(0,1/m)$
distributed entries. Using the Cauchy-Schwarz inequality twice, we get
\begin{align*} 
\bEE{\sum_{i=1}^{\nm}\frac{(a_{i}-b_{i})^{2}}{a_{\nm}b_{m}}}
\le
\sqrt{\bEE{\frac{1}{a_{\nm}^{2}b_{\nm}^{2}}}}\sqrt{\bEE{\bigg(\sum_{i=1}^{\nm}(a_{i}-b_{i})^{2}\bigg)^2}}\\
\le
\sqrt{\bEE{\frac{1}{a_{\nm}^{4}}}}\sqrt{\bEE{\bigg(\sum_{i=1}^{\nm}(a_{i}-b_{i})^{2}\bigg)^2}},
\end{align*}
where in the second inequality we also used the fact that $a_i$ and $b_i$ are identically distributed. The Hoffman-Wielandt inequality\footnote{{For normal matrices $A$ and $B$ with spectra $\{a_{i}\}$ and $\{b_{i}\}$, there exists a permutation $\pi$ of $[n]$ such that $\sum_{i}(a_{\pi(i)}-b_{i})^{2}\le\|A-B\|_{F}^{2}$. When $A$ and $B$ are p.s.d, the left-side is minimum when both sets of eigenvalues are arranged in increasing (or decreasing) order.} }~\cite{Hoffman_1953} can be used to handle the second term on the right-side. In particular, we have
$\sum_{i=1}^{\nm}(a_{i}-b_{i})^{2}\leq
\|A_{S_0}-A_{S_u}\|_{F}^{2}$ where the right-side coincides with
$\|\bphi_{1}\bphi_{1}^{\top}-\bphi_{\sps+1}\bphi_{\sps+1}^{\top}\|_{F}^{2}$ since $u=(1,k+1)$. 
Using the triangle inequality for Frobenius norm and noting that
$\|\bphi_i\bphi_i^\top\|_F$ equals $\|\bphi_i\|_2^2$ for 
a vector $\bphi_i$, we get
\begin{align*}
{\bEE{\sum_{i=1}^{\nm} \frac{(a_{i}-b_{i})^{2}}{a_{m} b_{m}}} }
&\leq
 \sqrt{\bEE{\frac{1}{a_{\nm}^{4}}}}\sqrt{\bEE{(\|\bphi_{1}\|_{2}^{2}+\|\bphi_{\sps+1}\|_{2}^{2})^{4}}}.
\end{align*}
Recall that $\bphi_{1}$ and $\bphi_{\sps+1}$ are independent $\cN(0,\frac{1}{\nm}I)$ distributed random vectors, and therefore $\nm(\|\bphi_{1}\|_{2}^{2}+\|\bphi_{\sps+1}\|_{2}^{2})$ is a chi-squared random variable with $2\nm$ degrees of freedom. Using the expression for the fourth moment of a chi-squared random variable gives us 
\begin{align*}
{\bEE{\sum_{i=1}^{\nm} \frac{(a_{i}-b_{i})^{2}}{a_{m} b_{m}}} }
&\leq \sqrt{\bEE{\frac{1}{a_{\nm}^{4}}}}\sqrt{{\frac{1}{\nm^{4}}\frac{(\nm+3)!}{(\nm-1)!}}}\\
&\leq c^{\prime} \sqrt{\bEE{\frac{1}{a_{\nm}^{4}}}}
\end{align*}
where $c^{\prime}$ is an absolute constant.

It only remains to bound  $\bEE{1/a_{\nm}^{4}}$, where $a_{\nm}$ is the
minimum eigenvalue of the $(\nm \times \nm)$ Wishart
matrix $A_{S_u}$. Using Lemma \ref{lem_smin} in Appendix \ref{app_smin}, we can obtain
\begin{equation*}
\bEE{a_{\nm}^{-4}}
\le\frac{c^{\prime\prime}\nm^4}{\sps^4(1-\nm/\sps)^8}.
\end{equation*}
By combining all the steps above, we get 
\begin{equation*}
\frac 13 \geq \bPr{\hat{S}\neq S_U}\geq 1 - \frac {\frac{c\ns\nm^2}{\sps^2(1-\nm/\sps)^4} +1} {\log
	\sps(\dmn-\sps)},
\end{equation*}
for a constant $c$. Observing that the $(1-\nm/\sps)^{4}$ term can be absorbed into $c$ when $\nm<\sps/2$ yields the desired bound.


\section{{Extension to nonbinary variances}}\label{sec:nonbinary}
{In this section, we extend our results to the case where $\lambda$ is not necessarily binary. Specifically, we
have the following assumption.
\begin{assumption}\label{assump_3}
	The entries of $\bx_{i}$, $i\in[\ns]$, are independent and zero-mean with 
	$\bEE{\bx_{i,j}^2}=\lambda_j$ for $\lambda\in  \{u\in\mR^{d}:\|u\|_{0}=k,\lmin\le u_{i}\le \lmax  \}$ and  
	$\bx_{i,j}\sim\sbg(c \lambda_j)$, where $0<\lmin\le\lmax$ and $c$ is an absolute constant. In addition, {we assume that} $\lmin/\lmax>k/(k+m-1)$.
\end{assumption}
 
Our sample complexity result continues to hold with an additional scaling by a factor of $\lmax^{2}/\lmin^{2}$. 
In particular, we have the following result, where we restrict to the noiseless setting.
\begin{theorem}\label{thm_nonbin}
For $\sigma^2=0$, the sample complexity of support recovery under Assumptions~\ref{assump_2} and~\ref{assump_3} satisfies
\begin{equation*}
 n^{*}(m,k,d)\le C \frac{\lmax^{2}}{\lmin^{2}}\bigg(\frac{\sps}{\nm}+1+\frac{\sigma^{2}}{\lmax}\bigg)^{2}\log\bigg(\frac{\sps(\dmn-\sps)}{\delta}\bigg),
\end{equation*}  
provided $m\ge (\log k)^{2}$, where $C$ is an absolute constant. Further, when $\sigma^{2}=0$, we have the lower bound
\begin{align*}
n^{*}(m,k,d)\ge 
c\frac{\lmax^{2}}{\lmin^{2}}\frac{k^{2}}{m^{2}}\log(d-k+1),
\end{align*}
where $c$ is an absolute constant.
\end{theorem}}
The techniques used for proving the upper and lower bounds remain essentially the same, and we highlight the key changes in the next two subsections.

\subsection{Upper bound}
We start by extending the bias calculation in Lemma \ref{lem_bias} to the more general nonbinary setting. We omit the proof since it follows from straightforward calculations. 

\begin{lemma}\label{lem_bias_nonbin}
	Let the estimator $\tilde{\lambda}$ be as defined in \eqref{lambda_tilde}. Then, under Assumptions \ref{assump_1} and \ref{assump_2} with $c=c^{\prime}=c^{\prime\prime}=1$, we have that	
	\begin{align*}
	\bEE{\tilde{\lambda}_{i}}
	=\frac{\nm+1}{\nm}\lambda_{i}+\frac{1}{\nm}\text{Tr}(K_{\lambda})+\sigma^2,\quad i\in[d],
	\end{align*}
	where the expectation is with respect to the joint distribution of
	$(\bx_{1}^{\ns},\bphi_{1}^{\ns},\bw_{1}^{\ns})$.
\end{lemma}

	Note that, when $\lambda$ is binary, $\text{Tr}(K_{\lambda})=k$ and the result above reduces to Lemma \ref{lem_bias}.

Our final estimate for the support is the same as before, namely, it computes $\tilde{\lambda}$ and declares the indices of the $\sps$ largest entries as the support. However, as before, we work with a threshold based estimator, with the bias terms in Lemma \ref{lem_bias_nonbin} being accounted for in the threshold.

Following the same series of arguments as in the binary case  
with $\sigma^{2}=0$, and using the assumption that $\lambda_{i}\in [\lmin,\lmax]$, we have that $\bPr{\tilde{S}\neq S}\le\delta$ can be achieved provided that the following condition holds for every $i\in S$ and every $i^{\prime}\in S^{c}$:
\begin{align*}\label{sep-nonbin}
\frac{\lmin}{\ns}\sum_{j=1}^{\ns}\alpha_{ji}^2-\frac{\lmax}{\ns}\sum_{j=1}^{\ns}\alpha_{ji^\prime}^2
\ge&
\lmax\bigg( \max\bigg\{\sqrt{\frac{c_1}{\ns^2}\sum_{j=1}^{\ns}\alpha_{ji}^4\log\frac{1}{\delta^{\prime}}},\frac{c_2}{\ns}\max_{j\in[\ns]}\alpha_{ji}^2\log\frac{1}{\delta^{\prime}}\bigg\}
\nonumber\\
&+
\max\bigg\{\sqrt{\frac{c_1}{\ns^2}\sum_{j=1}^{\ns}\alpha_{ji^{\prime}}^4\log\frac{1}{\delta^{\prime}}},\frac{c_2}{\ns}\max_{j\in[\ns]}\alpha_{ji^{\prime}}^2\log\frac{1}{\delta^{\prime}}\bigg\}\bigg),
\end{align*}
where $\delta^{\prime}=\delta/(2~\text{max}\{k,d-k\})$.
Incorporating the scaling due to $\lmin$ and $\lmax$ into our concentration bounds in the proof of Lemma \ref{lem:sep}, and simplifying, we get that
\begin{equation*}
\ns\ge C~\frac{\lmax^{2}}{\lmin^{2}}\bigg(\frac{k}{m}+1\bigg)^{2}\log\bigg(\frac{\sps(\dmn-\sps)}{\delta}\bigg)
\end{equation*}
samples suffice for $\bPr{\tilde{S}\neq S}\le\delta$, provided $\lmin/\lmax>k/(k+m-1)$.


\subsection{Lower bound}
We assume that the unknown $\lambda$ is uniformly distributed over the set $\{\lambda^{(0)}, \lambda^{(1)},\ldots,\lambda^{(\dmn-\sps)}\}$, with $\lambda^{(i)}\in\mR^{d}$.
The $j$th entry of $\lambda^{(i)}$, denoted $\lambda_{j}^{(i)}$, is given by
\begin{align*}
\lambda_{j}^{(i)}=\begin{cases}
\lmax,\quad\text{if}~j\in[\sps-1],\\
\lmin,\quad\text{if}~j=\sps+i,\\
0,\quad\text{otherwise},
\end{cases}
\end{align*}
for any $i\in\{0,1,\ldots,\dmn-\sps\}$.

Our goal is to characterize the KL divergence between distributions on the measurements arising from two different $\lambda$s in the set we described above, one of which we fix as $\lambda^{(0)}$. Computing this divergence as before, we see that
\begin{align}\label{div-bd-nonbin}
D\big(\bPP{Y|\lambda}\|\bPP{Y|\lambda^{(0)}}\big)
\leq \frac {1}{2}\,\bEE{\sum_{i=1}^{\nm}\frac{(a^{\prime}_{i}-b^{\prime}_{i})^{2}}{a^{\prime}_{\nm}b_{m}^{\prime}}},
\end{align} 
where $\{a_{i}^{\prime}\}_{i=1}^{\nm}$ and $\{b_{i}^{\prime}\}_{i=1}^{\nm}$ denote the eigenvalues of $A_{\lambda}\ed\bphi K_{\lambda}\bphi^{\top}$ and $A_{\lambda^{(0)}}\ed\bphi K_{\lambda^{(0)}}\bphi^{\top}$ respectively. Noting that $\sum_{i=1}^{m}(a_{i}^{\prime}-b_{i}^{\prime})^{2}\le \|A_{\lambda}-A_{\lambda^{(0)}}\|_{F}^{2}=\lmin^{2}\|\bphi_{1}\bphi_{1}^{\top}-\bphi_{\sps+1}\bphi_{\sps+1}^{\top}\|_{F}^{2}$, an application of the Hoffman-Wielandt inequality yields
\begin{align}\label{eval-bd-nonbin}
\bEE{\sum_{i=1}^{\nm}\frac{(a^{\prime}_{i}-b^{\prime}_{i})^{2}}{a^{\prime}_{\nm}b_{m}^{\prime}}}\le
c\lmin^{2}\sqrt{\bEE{\frac{1}{(a_{m}^{\prime})^{4}}}}.
\end{align}
Recall that from Lemma \ref{lem_smin}, we have a bound on the fourth moment of the smallest eigenvalue $a_{m}$ of $A_{S}=\bphi_{S}\bphi_{S}$ for $S\subseteq[\dmn]$. We now try to relate $a_{m}^{\prime}$ and $a_{m}$. We start by noting that
\begin{align*}
A_{\lambda}&=\lmax\sum_{i=1}^{k-1}\bphi_{i}\bphi_{i}^{\top}+\lmin\bphi_{\sps+1}\bphi_{\sps+1}^{\top}\\
&\succcurlyeq \lmax\sum_{i=1}^{k-1}\bphi_{i}\bphi_{i}^{\top},
\end{align*}
{where $A\succcurlyeq B$ if $A-B$ is a positive semi-definite matrix.}
The above inequality in turn gives a relation between the eigenvalues of $A_{\lambda}$ and those of $\lmax\sum_{i=1}^{k-1}\bphi_{i}\bphi_{i}^{\top}$. In particular, for the minimum eigenvalue, we have $a_{m}^{\prime}\ge \lmax a_{m}$. Combining this fact with the inequalities in \eqref{div-bd-nonbin} and \eqref{eval-bd-nonbin}, and using Lemma~\ref{lem_smin}, we get
\begin{align*}
D\big(\bPP{Y|\lambda}\|\bPP{Y|\lambda_{0}}\big)
&\le c^{\prime}\frac{\lmin^{2}}{\lmax^{2}}\sqrt{\bEE{\frac{1}{a_{m}^{4}}}}\\
&\le c^{\prime\prime}\frac{\lmin^{2}}{\lmax^{2}}\frac{\nm^{2}}{(\sps-1)^{2}}\bigg(1-\frac{\nm}{\sps}\bigg)^{-4}.
\end{align*}
This is the same bound as in the binary case, except for an additional scaling by a factor of $\lmin^{2}/\lmax^{2}$. As a consequence of this, we can show, using similar calculations as before, that if 
\begin{equation*}
\ns \le c\frac{\lmax^{2}}{\lmin^{2}}\frac{\sps^{2}}{m^{2}}\bigg(1-\frac{\nm}{\sps}\bigg)^{4}\log(\dmn-\sps+1),
\end{equation*}
then the error probability $\bPr{\hat{S}\neq \mathtt{supp}(\lambda_{0})}\ge 1/3$.

\section{Concluding Remarks}\label{sec:discussion}
Our sample complexity result implies that independent measurements
applied to the same sample are much more useful for support recovery than those applied to independent samples. There are several possible extensions of our results.  

We have shown that covariance-based methods provide a reliable way to  recover the support in the measurement-constrained setting of $m<k$, where traditional sparse recovery methods do not work. This framework can accommodate other kinds of structures on the covariance matrix; one can simply project the closed form estimate onto an appropriate constraint set and then use this estimate to make inferences about the data.
It would be interesting to explore this in more detail.

One can consider using the same measurement matrix
for all samples. In this case, we observe empirically that our
estimate does not perform well, but we do not have a complete theoretical
understanding of this phenomenon for our setting (see \cite[Proposition
2]{Azizyan_TIT_2018} for a related discussion).



For the case where the samples do not share a common support but instead have supports drawn from a small set of allowed supports, our current estimator can be used to recover the union of the allowed supports. Designing estimators for segregating the union into individual supports and labeling the samples based on the supports is an interesting direction for further work.




\begin{appendices}

\section{Proof of Lemma \ref{lem:sep}}\label{app_pf_sep}
We recall the statement of Lemma \ref{lem:sep} here for easy reference.
\begin{lemma-app}\label{l:separation_fxd_index_re}
	For all pairs $(i,i^{\prime})\in S\times S^{c}$, the separation condition 
	\begin{align}
	\frac{1}{\ns}\sum_{j=1}^{\ns}\alpha_{ji}^2-\frac{1}{\ns}\sum_{j=1}^{\ns}\alpha_{ji^\prime}^2
	\ge&
\max\bigg\{\sqrt{\frac{c_1}{\ns^2}\sum_{j=1}^{\ns}\alpha_{ji}^4\log\frac{1}{\delta^{\prime}}},\frac{c_2}{\ns}\max_{j\in[\ns]}\alpha_{ji}^2\log\frac{1}{\delta^{\prime}}\bigg\}
\nonumber\\
&+
\max\bigg\{\sqrt{\frac{c_1}{\ns^2}\sum_{j=1}^{\ns}\alpha_{ji^{\prime}}^4\log\frac{1}{\delta^{\prime}}},\frac{c_2}{\ns}\max_{j\in[\ns]}\alpha_{ji^{\prime}}^2\log\frac{1}{\delta^{\prime}}\bigg\}\numberthis\label{sep_re}
	\end{align}
	holds with probability at
	least $1-\delta$ if $n\geq c (\sps/\nm +\sigma^2)^2\log(1/\delta^{\prime})$ 
{and $m\ge (\log k)^{2}$}, where $\delta^{\prime}=\delta/(4~\text{max}\{k,d-k\})$.
\end{lemma-app} 
\begin{proof}
The proof involves studying the tail behaviour of each term in \eqref{sep_re}. In particular, we derive a lower bound on the first term and upper bounds on the remaining terms that hold with high probability over the subgaussian measurement ensemble, and establish conditions under which the separation in \eqref{sep_re} holds for a fixed pair $(i,i^{\prime})$. A union bound over all $k(d-k)$ pairs then gives us the result claimed in the lemma. 
The key technicality is to keep track of the leading term that is contributed by each of the terms in \eqref{sep_re}. To get a rough idea of the behaviour of these terms, recall that $\alpha_{ji}$ depends on inner products between the columns $\Phi_{j}$. In particular, it involves the sum of $O(k)$ inner product squared terms, and this scales as $O(k/m)$ in expectation. The right side in \eqref{sep_re} thus effectively leads to a $O(1)$ term. The left side, on the other hand, gives rise to a $O(\sqrt{1/n\cdot (k/m)^{2}})$ term in expectation, leading to the overall requirement of $n=O(k^{2}/m^{2})$. In what follows, we make these arguments precise using tail bounds for each of the terms in \eqref{sep_re}.

For clarity of presentation, details of the tail bounds for each term in \eqref{sep_re} are presented in Appendix~\ref{app_lemmas}, which in turn build on standard concentration bounds for subgaussian and subexponential random variables reviewed as preliminaries in Appendix \ref{app_prelim}. Also, while analyzing each term in \eqref{sep_re}, we use the same symbol $\mu$ to denote the expectation of that term to keep notation simple. Similarly, the definitions of terms like $\mu_{1}$, $\mu_{2}$ and $\mu_{3}$ will be clear from the context.

 For the first term on the left side of \eqref{sep_re}, we study the behaviour of its left tail. That is, we look at $\bPr{{\frac{1}{\ns}\sum\limits_{j=1}^{\ns}\alpha_{ji}^{2}\le\mu-t}}$, where recall
\begin{equation}\label{alph2}
\alpha_{ji}^{2}=\|\bphi_{ji}\|_{2}^{4}+\sum_{l\in S\backslash\{i\}}(\bphi_{jl}^{\top}\bphi_{ji})^{2}+\sigma^2\|\bphi_{ji}\|_{2}^{2},
\end{equation}
and
$\mu=\bEE{(1/n)\sum_{j=1}^{\ns}\alpha_{ji}^2}.$ Further, let $\mu_{1}, \mu_{2}$ and $\mu_{3}$ denote the mean of each of the three terms.
 By a union bound argument, it suffices to bound the normalized sum of each of the three terms separately. Notice that all these terms essentially depend on the lengths of the columns or the inner products between the columns of the measurement matrix, and our goal will be obtain concentration bounds for these terms. While $\|\bphi_{ji}\|_{2}^{2}$ is clearly subexponential, $\|\bphi_{ji}\|_{2}^{4}$ and $(\bphi_{jl}^{\top}\bphi_{ji})^{2}$ have heavier tails.
 In Appendix \ref{app_lemmas}, we provide results on the tail behaviour of these terms. Using Lemma \ref{lem_z4left}, we have for any $t>0$,
\begin{align*}
\bPr{{\frac{1}{\ns}\sum_{j=1}^{\ns}\|\bphi_{ji}\|_{2}^{4}\le \mu_1-t}}\le\frac{\ep}{3}
\end{align*}
for
\begin{align*}
\mu_{1}-t=\text{min}\bigg\{\bigg(1-\sqrt{\frac{c_1}{\nm\ns}\log\frac{3}{\varepsilon}}\bigg)^{2},\bigg(1-\frac{c_2}{\nm\ns}\log\frac{3}{\varepsilon}\bigg)^{2}\bigg\}.
\end{align*}
Further, from Lemma \ref{lem_ip_conc}, we have that when $n\ge(c_{2}^{2}/{c_{1}})\log(12/\varepsilon)$,
\begin{align*}
\bPr{{\frac{1}{\ns}\sum_{j=1}^{\ns}\sum_{t\in S\backslash\{i\}}(\bphi_{jt}^{\top}\bphi_{ji})^{2}\le \mu_2-t}}\le\frac{\ep}{3}
\end{align*}
for
\begin{align*}
\mu_{2}-t=&\frac{k-1}{\nm}\bigg(1-\sqrt{\frac{c_1}{\nm\ns}\log\frac{12}{\varepsilon}}\bigg)\\
&- \sqrt{\frac{1}{\nm\ns}\log\frac{12}{\varepsilon}}~\text{max}\bigg\{\sqrt{
	c_{1}\frac{k-1}{m}},c_{2}\bigg\}\text{max}\bigg\{\bigg(1+\sqrt{\frac{c_1}{\nm}\log\frac{12\ns}{\varepsilon}}\bigg),\bigg(1+{\frac{c_1}{\nm}\log\frac{12\ns}{\varepsilon}}\bigg)\bigg\}.
\end{align*}
Finally, from Lemmas \ref{lem_sbx_sbg} and \ref{lem_sbx_txf}, we can see that $\|\bphi_{ji}\|_{2}^{2}$ is subexponential with parameters $(c_{1}/\nm,c_{2}/\nm)$ and that $(\sigma^{2}/\ns)\sum_{j=1}^{\ns}\|\bphi_{ji}\|_{2}^{2}$ is subexponential with parameters $(c_{1}\sigma^{4}/\nm\ns,c_{2}\sigma^{2}/\nm\ns)$. Using the subexponential concentration bound from Lemma \ref{lem_z2leftright} gives
\begin{align*}
\bPr{{\frac{\sigma^2}{\ns}\sum_{j=1}^{\ns}\|\bphi_{ji}\|_{2}^{2}\le \mu_3-t}}\le\frac{\ep}{3}
\end{align*} 
for
\begin{align*}
\mu_{3}-t=\sigma^2\bigg(1-\text{max}\bigg\{\sqrt{\frac{c_1}{\nm\ns}\log\frac{3}{\varepsilon}},\frac{c_2}{\nm\ns}\log\frac{3}{\varepsilon}\bigg\}\bigg).
\end{align*}

Combining these results using a union bound step, we see that
\begin{align*}
\bPr{\frac{1}{\ns}\sum_{j=1}^{\ns}\alpha_{ji}^{2}\le \mu-t}\le \varepsilon,
\end{align*}
for
\begin{align*}\label{conc_term1}
\mu-t=&~\bigg(1-\sqrt{\frac{c_1}{\nm\ns}\log\frac{3}{\varepsilon}}\bigg)^{2}+\frac{k-1}{\nm}\bigg(1-\sqrt{\frac{c_1}{\nm\ns}\log\frac{12}{\varepsilon}}\bigg)\nonumber\\
&
- \sqrt{\frac{1}{\nm\ns}\log\frac{12}{\varepsilon}}~\text{max}\bigg\{\sqrt{
	c_{1}\frac{k-1}{m}},c_{2}\bigg\}\text{max}\bigg\{\bigg(1+\sqrt{\frac{c_1}{\nm}\log\frac{12\ns}{\varepsilon}}\bigg),\bigg(1+{\frac{c_2}{\nm}\log\frac{12\ns}{\varepsilon}}\bigg)\bigg\}\nonumber\\
&+\sigma^2\bigg(1-\sqrt{\frac{c_1}{\nm\ns}\log\frac{3}{\varepsilon}}\bigg),
\end{align*}
when $n\ge(c_{2}^{2}/{c_{1}})\log(12/\varepsilon)$.


We now consider the second term on the left side of  \eqref{sep_re}, and observe that it consists of terms similar to the ones we encountered in the previous calculation. Our focus will be on the \emph{right} tail this time, i.e., we will study $\bPr{\frac{1}{\ns}\sum\limits_{j=1}^{\ns}\alpha_{ji^\prime}^{2}\ge \mu+t}$ for $i^\prime\in S^{c}$, where
\begin{equation*}
\alpha_{ji^\prime}^{2}=\sum_{l\in S}(\bphi_{jl}^{\top}\bphi_{ji^\prime})^{2}+\sigma^2\|\bphi_{ji^\prime}\|_{2}^{2}.
\end{equation*} 
We use Lemma \ref{lem_z2leftright} to get
\begin{align*}
\bPr{{\frac{\sigma^2}{\ns}\sum_{j=1}^{\ns}\|\bphi_{ji^\prime}\|_{2}^{2}\ge \mu_2+t_{2}}}\le\frac{\ep}{2}
\end{align*} 
for 
\begin{align*}
\mu_{2}+t_{2}=\sigma^2\bigg(1+\text{max}\bigg\{\sqrt{\frac{c_1}{\nm\ns}\log\frac{2}{\varepsilon}},\frac{c_2}{\nm\ns}\log\frac{2}{\varepsilon}\bigg\}\bigg),
\end{align*}
and Lemma \ref{lem_ip_conc_2} to get
\begin{align*}
\bPr{{\frac{1}{\ns}\sum_{j=1}^{\ns}\sum_{l\in S}(\bphi_{jl}^{\top}\bphi_{ji^\prime})^{2}\ge \mu_1+t_{1}}}
\le\frac{\ep}{2}
\end{align*}
for
\begin{align*}
\mu_{1}+t_{1}=&\frac{\sps}{\nm}\bigg(1+\sqrt{\frac{c_1}{\nm\ns}\log\frac{8}{\varepsilon}}\bigg)\\
&+ \sqrt{\frac{1}{\nm\ns}\log\frac{8}{\varepsilon}}~\text{max}\bigg\{\sqrt{
	c_{1}\frac{k}{m}},c_{2}\bigg\}\text{max}\bigg\{\bigg(1+\sqrt{\frac{c_1}{\nm}\log\frac{8\ns}{\varepsilon}}\bigg),\bigg(1+{\frac{c_2}{\nm}\log\frac{8\ns}{\varepsilon}}\bigg)\bigg\},
\end{align*}
when $n\ge (c_{2}^{2}/{c_{1}})\log(8/\varepsilon)$.
Putting these results together, we get
\begin{align*}
\bPr{{\frac{1}{\ns}\sum_{j=1}^{\ns}\alpha_{ji^\prime}^{2}\ge \mu+t}}\le \varepsilon,
\end{align*}
for
\begin{align*}
\mu+t
=&~\frac{\sps}{\nm}\bigg(1+\sqrt{\frac{c_1}{\nm\ns}\log\frac{8}{\varepsilon}}\bigg)\\
&+ \sqrt{\frac{1}{\nm\ns}\log\frac{8}{\varepsilon}}~\text{max}\bigg\{\sqrt{
	c_{1}\frac{k}{m}},c_{2}\bigg\}\text{max}\bigg\{\bigg(1+\sqrt{\frac{c_1}{\nm}\log\frac{8\ns}{\varepsilon}}\bigg),\bigg(1+{\frac{c_2}{\nm}\log\frac{8\ns}{\varepsilon}}\bigg)\bigg\}\nonumber\\
&+\sigma^2\bigg(1+\sqrt{\frac{c_1}{\nm\ns}\log\frac{2}{\varepsilon}}\bigg),
\end{align*}
when $n\ge (c_{2}^{2}/{c_{1}})\log(8/\varepsilon)$.

For the third term in \eqref{sep_re}, namely, $\max\bigg\{\sqrt{\frac{c_1}{\ns^2}\sum_{j=1}^{\ns}\alpha_{ji}^4\log\frac{1}{\delta^{\prime}}},\frac{c_2}{\ns}\max_{j\in[\ns]}\alpha_{ji}^2\log\frac{1}{\delta^{\prime}}\bigg\}$, we consider the possibility of either argument attaining the maximum and study the respective right tails.

First, we look at  $\bPr{{\sqrt{\frac{1}{\ns^{2}}\sum\limits_{j=1}^{\ns}\alpha_{ji}^{4}}\ge\mu+t}}$ for $i\in S$. We note that by the union bound,
\begin{align}
\bPr{{\sqrt{\frac{1}{\ns^{2}}\sum\limits_{j=1}^{\ns}\alpha_{ji}^{4}}\ge\mu+t}}
\le&~\sum_{j=1}^{\ns}\bPr{\alpha_{ji}^4\ge\ns(\mu+t)^2}\nonumber \\
\le&~ \ns\bPr{\|\bphi_{1i}\|_{2}^{4}\ge\frac{\sqrt{\ns}}{3}(\mu+t)}
+\ns\bPr{\sum_{l\in S\backslash\{i\}}(\bphi_{1i}^{\top}\bphi_{1l})^2\ge\frac{\sqrt{\ns}}{3}(\mu+t)}\nonumber\\
&+\ns\bPr{{\sigma^2\|\bphi_{1i}\|_{2}^{2}\ge\frac{\sqrt{\ns}}{3}(\mu+t)}}.\label{conc_term3_unb}
\end{align}
We use Lemma \ref{lem_sbx_tail} for the first and third terms and Lemma \ref{lem_ip_conc_2} for the second term. A direct application of Lemma \ref{lem_ip_conc_2} with $n=1$ for the second term however requires the assumption that {$m\ge (c_{2}^{2}/{c_{1}})\log(12 n/\varepsilon)$} (note that the second term in \eqref{conc_term3_unb} needs to be upper bounded by $\varepsilon/3n$). While in our setting such an assumption on $n$ is acceptable, we would like to avoid making this assumption on $m$ at this stage.
 We therefore omit the simplification done at the end of Lemma \ref{lem_ip_conc_2} 
 to get
\begin{equation*}
\bPr{{\sqrt{\frac{1}{\ns^{2}}\sum\limits_{j=1}^{\ns}\alpha_{ji}^{4}}\ge\mu+t}}\le\varepsilon
\end{equation*}
for
\begin{align*}\label{conc_term3}
\mu+t=&~ \frac{3}{\sqrt{\ns}}\bigg(1+\text{max}\bigg\{\sqrt{\frac{c_1}{\nm}\log\frac{3n}{\varepsilon}},\frac{c_2}{\nm}\log\frac{3n}{\varepsilon}\bigg\}\bigg)^{2}\\
&+\frac{3}{\sqrt{n}}\bigg(\frac{k-1}{m}+\text{max}\bigg\{\frac{c_{2}}{m}\log\frac{9n}{\varepsilon},\sqrt{c_{1}\frac{k-1}{m^{2}}\log\frac{9n}{\varepsilon}}\bigg\}\bigg)\bigg(1+\max\bigg\{\sqrt
{\frac{c_{1}}{m}\log\frac{9n}{\varepsilon}},\frac{c_{2}}{m}\log\frac{9n}{\varepsilon}\bigg\}\bigg)\\
&+\frac{3\sigma^2}{\sqrt{\ns}}\bigg(1+\text{max}\bigg\{\sqrt{\frac{c_1}{\nm}\log\frac{3\ns}{\varepsilon}},\frac{c_2}{\nm}\log\frac{3\ns}{\varepsilon}\bigg\}\bigg)\numberthis.
\end{align*}

Next, we look at  $\bPr{{\max_{j\in[\ns]}\alpha_{ji}^2\ge\mu+t}}$ for $i\in S$. We notice that by the union bound, we have
\begin{align*}
\bPr{{\max_{j\in[\ns]}\alpha_{ji}^2\ge\mu+t}}
&\le \sum_{j=1}^{n}\bPr{{\alpha_{ji}^2\ge \mu+t}}\\
&\le\sum_{j=1}^{n}\bigg[\bPr{{\|\bphi_{ji}\|_{2}^{4}\ge\frac{\mu+t}{3}}}+\bPr{{\sum_{l\in S\backslash\{i\}}(\bphi_{ji}^{\top}\bphi_{jl})^2\ge \frac{\mu+t}{3}}}\nonumber\\
&~~~~~~~~~+\bPr{{\sigma^{2}\|\bphi_{ji}\|_{2}^{2}\ge\frac{\mu+t}{3}}}\bigg].\numberthis\label{term3_max}
\end{align*}
 We now handle each of the three terms on the right-side of \eqref{term3_max} separately. We will use Lemma \ref{lem_z2leftright} for the first and third terms and Lemma \ref{lem_ip_conc} for the second term. 
 \begin{onlyarxiv}
 In particular, for every $j\in[n]$, we have that
\begin{align*}
\bPr{{\|\bphi_{ji}\|_{2}^{4}\ge\frac{\mu+t}{3}}}
&\le \varepsilon,
\end{align*}
for 
\begin{align*}
\mu+t=3\bigg(1+\max\bigg\{\sqrt{
\frac{c_{1}}{m}\log\frac{1}{\varepsilon}},\frac{c_{2}}{m}\log\frac{1}{\varepsilon}\bigg\}\bigg)^{2},
\end{align*}
 and that
\begin{align*}
\bPr{{\|\bphi_{1i}\|_{2}^{2}\ge\frac{\mu+t}{3\sigma^{2}}}}
&\le \varepsilon
\end{align*}
for 
\begin{align*}
\mu+t=3\sigma^{2}\bigg(1+\max\bigg\{\sqrt{
\frac{c_{1}}{m}\log\frac{1}{\varepsilon}},\frac{c_{2}}{m}\log\frac{1}{\varepsilon}\bigg\}\bigg).
\end{align*}
For the second term, we have that for every $j\in[n]$,
\begin{align*}
\bPr{{\sum_{l\in S\backslash\{i\}}(\bphi_{ji}^{\top}\bphi_{jl})^2\ge \frac{\mu+t}{3}}}
&\le \varepsilon,
\end{align*}
for
\begin{align*}
\mu+t=3\bigg(\frac{k-1}{m}+\sqrt{c_{2}\frac{k-1}{m^{2}}\log\frac{3}{\varepsilon}}\bigg)\bigg(1+\max\bigg\{\sqrt
{\frac{c_{1}}{m}\log\frac{3}{\varepsilon}},\frac{c_{2}}{m}\log\frac{3}{\varepsilon}\bigg\}\bigg).
\end{align*}
\end{onlyarxiv}
Substituting those bounds into \eqref{term3_max}, we get
\begin{align*}
\bPr{{\max_{j\in[\ns]}\alpha_{ji}^2\ge\mu+t}}
&\le 3n\varepsilon,
\end{align*}
for
\begin{equation*}
\mu+t=3(1+f(m,\varepsilon))\cdot\max\bigg\{\sigma^{2},1+f(m,\varepsilon),\frac{k-1}{m}+\sqrt{c_{2}\frac{k-1}{m^{2}}\log\frac{3}{\varepsilon}}\bigg\}
\end{equation*} 
where $f(m,\varepsilon)=\max\{\sqrt
{\frac{c_{1}}{m}\log\frac{3}{\varepsilon}},\frac{c_{2}}{m}\log\frac{3}{\varepsilon}\}$. That is,
\begin{equation*}
\bPr{\frac{1}{n}{\max_{j\in[\ns]}\alpha_{ji}^2\ge\mu+t}}\le\varepsilon
\end{equation*}
for
\begin{equation}\label{term3_max_2}
\mu+t=\frac{3}{n}(1+f(m,\varepsilon/3n))\cdot\max\bigg\{\sigma^{2},1+f(m,\varepsilon/3n),\frac{k-1}{m}+\sqrt{c_{2}\frac{k-1}{m^{2}}\log\frac{9n}{\varepsilon}}\bigg\}.
\end{equation} 
{Comparing \eqref{conc_term3} and \eqref{term3_max_2}, we see that $\frac{1}{n}\max_{j\in[\ns]}\alpha_{ji}^2$ is $O(k/mn+\sigma^{2}/n)$ which decays faster with respect to $n$ compared to $\sqrt{\frac{1}{\ns^{2}}\sum\limits_{j=1}^{\ns}\alpha_{ji}^{4}}$, which is $O(k/m\sqrt{n}+\sigma^{2}/\sqrt{n})$. Thus, the third term in \eqref{sep_re} is dominated by the $O(k/m\sqrt{n}+\sigma^{2}/\sqrt{n})$ term, which is what we retain in our subsequent calculations.}

Finally, for the fourth term in \eqref{sep_re}, we first look at $\bPr{{\sqrt{({1}/{n^{2}})\sum\limits_{j=1}^{\ns}\alpha_{ji^\prime}^{4}}\ge \mu+t}}$ for $i^{'}\in S^{c}$. Using similar arguments as in the previous calculation, we get
\begin{equation*}
\bPr{{\sqrt{\frac{1}{\ns^{2}}\sum\limits_{j=1}^{\ns}\alpha_{ji^\prime}^{4}}\ge\mu+ t}}\le \varepsilon,
\end{equation*}
for 
\begin{align*}
\mu+t
&=\frac{2}{\sqrt{n}}\bigg(\frac{k}{m}+\text{max}\bigg\{\frac{c_{2}}{m}\log\frac{6n}{\varepsilon},\sqrt{c_{1}\frac{k}{m^{2}}\log\frac{6n}{\varepsilon}}\bigg\}\bigg)
\bigg(1+\max\bigg\{\sqrt
{\frac{c_{1}}{m}\log\frac{6n}{\varepsilon}},\frac{c_{2}}{m}\log\frac{6n}{\varepsilon}\bigg\}\bigg)\\
&+\frac{2\sigma^2}{\sqrt{\ns}}\bigg(1+\text{max}\bigg\{\sqrt{\frac{c_1}{\nm}\log\frac{2\ns}{\varepsilon}},\frac{c_2}{\nm}\log\frac{2\ns}{\varepsilon}\bigg\}\bigg).
\end{align*}
The $\frac{1}{n}\max_{j\in[\ns]}\alpha_{ji^{\prime}}^2$ term, as we discussed before, will lead to a $O(k/mn)$ factor, which can be ignored.

The foregoing calculations provide bounds on each of the four terms occuring in \eqref{sep_re}, that hold with high probability. We note that the left-side of \eqref{sep_re} is lower bounded by
\begin{align*}
1-\frac{1}{m}-2\sqrt{\frac{c_{1}}{mn}\log\frac{24}{\varepsilon}}\bigg(1+\sigma^{2}+\frac{k}{m}\bigg)+\frac{c_{1}}{mn}\log\frac{6}{\varepsilon}\\
-2\sqrt{\frac{c_{1}k}{m^{2}n}\log\frac{24}{\varepsilon}}\bigg(1+\frac{c_{2}}{m}\log\frac{24n}{\varepsilon}\bigg)\numberthis\label{term12lb}
\end{align*}
with probability at least $1-\varepsilon$, and that the right-side is upper bounded by
\begin{align*}
5\sqrt{\frac{c_{1}}{n}\log\frac{1}{\delta^{\prime}}}
&\bigg[
\bigg(1+\text{max}\bigg\{\sqrt{\frac{c_1}{\nm}\log\frac{6n}{\varepsilon}},\frac{c_2}{\nm}\log\frac{6n}{\varepsilon}\bigg\}\bigg)^{2}\\*
&+\bigg(\frac{k}{m}+\text{max}\bigg\{\frac{c_{2}}{m}\log\frac{18n}{\varepsilon},\sqrt{c_{1}\frac{k}{m^{2}}\log\frac{18n}{\varepsilon}}\bigg\}\bigg)\bigg(1+\max\bigg\{\sqrt
{\frac{c_{1}}{m}\log\frac{18n}{\varepsilon}},\frac{c_{2}}{m}\log\frac{18n}{\varepsilon}\bigg\}\bigg)\\
&+\sigma^2\bigg(1+\text{max}\bigg\{\sqrt{\frac{c_1}{\nm}\log\frac{6\ns}{\varepsilon}},\frac{c_2}{\nm}\log\frac{6\ns}{\varepsilon}\bigg\}\bigg)\bigg],\numberthis\label{term34ub}
\end{align*}
with probability at least $1-\varepsilon$. 
To ensure that \eqref{sep_re} holds with probability at least $1-\varepsilon$ for a fixed $(i,i^{\prime})\in S\times S^{c}$, we need that \eqref{term12lb} exceeds \eqref{term34ub}.
For further simplification, we assume $m$ to be sufficiently large to handle the $\log n$ terms. This assumption on $m$ can possibly be removed by handling the sum in Lemma \ref{lem_z4right} and \eqref{conc_term3_unb} directly and not using the union bound. Choosing $\varepsilon=\delta/(4~k(d-k))$ to account for the union bound over all $(i,i^{\prime})$ pairs and focusing on the  $n=O((k/m+1+\sigma^{2})^{2}\log(1/\delta^{\prime}))$ regime, we see that \eqref{term12lb} exceeds \eqref{term34ub} and separation holds 
if\footnote{{We use this condition 
		to show that $(1/\sqrt{m})\log (k/m)\leq 1$
		and the dominating term on the right-side of~\eqref{term34ub} 
		is $k/m$.}} $m\ge (\log k)^{2}$.

 Thus,
\begin{equation*}
n\geq c \bigg(\frac{\sps}{\nm} +1+\sigma^2\bigg)^2\log \frac {1}{\delta^{\prime}} 
\end{equation*}
samples suffice to ensure separation between the typical values and to guarantee that \eqref{sep_re} holds with probability at least $1-\delta$.

\end{proof}

\section{Key technical lemmas}\label{app_lemmas}
\begin{onlyjournal}
Proofs of the results in this section can be found in~\cite{Ramesh_arxiv_2019}.
\end{onlyjournal}

\begin{lemma-app} \label{lem_z2leftright}
	Let $Z_{1},\ldots,Z_{\ns}$ be independent, mean-zero random vectors in $\mR^\nm$ with independent {strictly subgaussian} entries with variance $1/\nm$. Then, there exist absolute constants $c_1$ and $c_2$ such that for any $t>0$,
	\begin{align*}
		\bPr{\bigg\rvert\frac{1}{\ns}\sum_{j=1}^{\ns}\|Z_{j}\|_{2}^{2}-1\bigg\rvert\ge t}\le 2\exp\bigg(-\text{min}\bigg\{\frac{\nm\ns}{c_1}t^{2},\frac{\nm\ns}{c_2}t\bigg\}\bigg).
	\end{align*} 
	Equivalently, for any $\epsilon>0$,
	\begin{equation*}
\bPr{\bigg\rvert\frac{1}{\ns}\sum_{j=1}^{\ns}\|Z_{j}\|_{2}^{2}-1\bigg\rvert\ge \text{max}\bigg\{\sqrt{\frac{c_1}{\nm\ns}\log\frac{2}{\varepsilon}},\frac{c_2}{\nm\ns}\log\frac{2}{\varepsilon}\bigg\}}
\le \varepsilon.
\end{equation*}
\end{lemma-app}

\begin{onlyarxiv}
\begin{proof}
	Since $Z_{jl}\sim\sbg(1/\nm)$ for any $j\in[\ns]$ and $l\in[\nm]$, we have from Lemma \ref{lem_sbx_sbg} that $Z_{jl}^{2}\sim\sbx(c_{1}/\nm^2,c_{2}/\nm)$ for some absolute constants $c_{1}$ and $c_{2}$. Using properties of subexponential random variables from Lemma \ref{lem_sbx_txf}, we can show that the normalized sum $\frac{1}{\ns}\sum_{j=1}^{\ns}\|Z_{j}\|_{2}^{2}$ is also subexponential with parameters $(c_1/\nm\ns,c_2/\nm\ns)$. Noting that $\bEE{\frac{1}{\ns}\sum_{j=1}^{\ns}\|Z_{j}\|_{2}^{2}}=1$ and using the tail bound from Lemma \ref{lem_sbx_tail} we get for $t>0$,
\begin{align}\label{sbx_tail_abs}
\bPr{\bigg\rvert\frac{1}{\ns}\sum_{j=1}^{\ns}\|Z_{j}\|_{2}^{2}-1\bigg\rvert \ge t}\le 2\exp\bigg(-\text{min}\bigg\{\frac{\nm\ns}{c_1}t^2,\frac{\nm\ns}{c_2}t\bigg\}\bigg).\numberthis
\end{align}
For the right side to be at most $\ep>0$, we see that it suffices to have
\begin{align*}
t\ge \text{max}\bigg\{\sqrt{\frac{c_1}{\nm\ns}\log\frac{2}{\ep}},\frac{c_2}{\nm\ns}\log\frac{2}{\ep}\bigg\}.
\end{align*}
Substituting the above into \eqref{sbx_tail_abs} gives us the result.	
\end{proof}
\end{onlyarxiv}

\begin{lemma-app} \label{lem_z4left}
	Let $Z_{1},\ldots,Z_{\ns}$ be independent, mean-zero random vectors in $\mR^\nm$ with independent strictly subgaussian entries with variance $1/\nm$. Then, there exist absolute constants $c_1$ and $c_2$ such that for any $\varepsilon>0$,
	\begin{equation*}
		\bPr{{\frac{1}{\ns}\sum_{j=1}^{\ns}\|Z_{j}\|_{2}^{4}\le\text{min}\bigg\{\bigg(1-\sqrt{\frac{c_1}{\nm\ns}\log\frac{1}{\varepsilon}}\bigg)^{2},\bigg(1-\frac{c_2}{\nm\ns}\log\frac{1}{\varepsilon}\bigg)^{2}\bigg\}}}\le\varepsilon.
	\end{equation*}
\end{lemma-app}

\begin{onlyarxiv}
\begin{proof}
	Let $\mu=\bEE{\frac{1}{\ns}\sum_{j=1}^{\ns}\|Z_{j}\|_{2}^{4}}=1+2/\nm$, and $t<\mu$. Then, using Jensen's inequality, we have
	\begin{align*}
		\bPr{{\frac{1}{\ns}\sum_{j=1}^{\ns}\|Z_{j}\|_{2}^{4}\le\mu-t}}
		&\le\bPr{{\frac{1}{\ns}\sum_{j=1}^{\ns}\|Z_{j}\|_{2}^{2}\le\sqrt{\mu-t}}}\\
		&=\bPr{{\frac{1}{\ns}\sum_{j=1}^{\ns}\|Z_{j}\|_{2}^{2}-1\le -t^\prime}},
	\end{align*}
	where $t^\prime=1-\sqrt{\mu-t}$. Using Lemma \ref{lem_z2leftright} and reparameterizing with respect to $\varepsilon>0$ gives the result.

\end{proof}
\end{onlyarxiv}

\begin{lemma-app}\label{lem_z4right}
	Let $Z_{1},\ldots,Z_{\ns}$ be independent, mean-zero random vectors in $\mR^\nm$ with independent strictly subgaussian entries with variance $1/\nm$. Then, there exist absolute constants $c_1$ and $c_2$ such that for any $\varepsilon>0$,
	\begin{equation*}
		\bPr{\frac{1}{\ns}\sum_{j=1}^{\ns}\|Z_{j}\|_{2}^{4}\ge\text{max}\bigg\{\bigg(1+\sqrt{\frac{c_1}{\nm}\log\frac{\ns}{\varepsilon}}\bigg)^{2},\bigg(1+{\frac{c_2}{\nm}\log\frac{\ns}{\varepsilon}}\bigg)^{2}\bigg\}}\le \varepsilon.
	\end{equation*}
\end{lemma-app}

\begin{onlyarxiv}
\begin{proof}
	
	Let $\mu=\bEE{\frac{1}{\ns}\sum_{j=1}^{\ns}\|Z_{j}\|_{2}^{4}}$ and note as we did in Lemma \ref{lem_z2leftright} that $\|Z_{j}\|_{2}^{2}\sim\sbx(c_{1}/\nm,c_{2}/\nm)$. We have	by union bound that
	\begin{align*}
		\bPr{\frac{1}{\ns}\sum_{j=1}^{\ns}\|Z_{j}\|_{2}^{4}\ge\mu+t}
		&\le \sum_{j=1}^{\ns}\bPr{\|Z_{j}\|_{2}^{2}-1\ge \sqrt{\mu+t}-1}\\
		&\le \ns \exp\bigg(-\text{min}\bigg\{\frac{\nm(t^\prime)^2}{c_1},\frac{\nm t^\prime}{c_2}\bigg\}\bigg)
	\end{align*}
	where the last inequality follows from Lemma \ref{lem_sbx_tail} with $t^\prime=\sqrt{\mu+t}-1$.
	Equating the expression on the right to $\varepsilon$ and reparameterizing gives the result.

\end{proof}
\end{onlyarxiv}
\begin{lemma-app}\label{lem_ip_conc}
	Let $Z_{j},Y_{j1},\ldots,Y_{j,k-1}$, $j\in[\ns]$, be independent, mean-zero random vectors in $\mR^\nm$ with independent strictly subgaussian entries with variance $1/\nm$. Let $\mu=\bEE{\frac{1}{\ns}\sum_{j=1}^{\ns}\sum_{l=1}^{k-1}(Y_{jl}^{\top}Z_{j})^{2}}$. Then, there exist absolute constants $c_1$ and $c_2$ such that for any $\varepsilon>0$,
	\begin{equation*}
		\bPr{\frac{1}{\ns}\sum_{j=1}^{\ns}\sum_{l=1}^{k-1}(Y_{jl}^{\top}Z_{j})^{2}\le\mu-t}\le\varepsilon,
	\end{equation*}
	for
	\begin{align*}
	\mu-t 
	=&~ \frac{\sps-1}{\nm}\bigg(1-\sqrt{\frac{c_1}{\nm\ns}\log\frac{4}{\varepsilon}}\bigg)\nonumber\\
	&-\sqrt{\frac{1}{\nm\ns}\log\frac{4}{\varepsilon}}~\text{max}\bigg\{\sqrt{
		c_{1}\frac{k-1}{m}},c_{2}\bigg\}\text{max}\bigg\{\bigg(1+\sqrt{\frac{c_{1}}{\nm}\log\frac{4\ns}{\varepsilon}}\bigg),\bigg(1+{\frac{c_{2}}{\nm}\log\frac{4\ns}{\varepsilon}}\bigg)\bigg\},
	\end{align*}

\end{lemma-app}
when $n\ge(c_{2}^{2}/{c_{1}})\log (4/\varepsilon)$.

\begin{onlyarxiv}
\begin{proof}
Note that conditioned on $Z_j$, the random variable $Y_{jl}^{\top}Z_{j}$ is subgaussian with parameter $\|Z_{j}\|_{2}^{2}/\nm$, for any $j\in[\ns]$ and $l\in[\sps-1]$. Using Lemmas \ref{lem_sbx_sbg} and \ref{lem_sbx_txf}, we have that the normalized sum $(1/\ns)\sum_{j=1}^{\ns}\sum_{l=1}^{k-1}(Y_{jl}^{\top}Z_{j})^{2}$, conditioned on $\{Z_{j}\}_{j=1}^{\ns}$, is subexponential with parameters $v^{2}$ and $b$ where
\begin{align*}
v^{2}=\frac{c_{1}}{\nm^2\ns^2}(k-1)\sum_{j=1}^{\ns}\|Z_{j}\|_{2}^{4},\quad b=\frac{c_2}{\nm\ns}\text{max}_{j\in[\ns]}\|Z_{j}\|_{2}^{2}.
\end{align*}
	Let $\mu^\prime=\bEE{\frac{1}{\ns}\sum_{j=1}^{\ns}\sum_{l=1}^{k-1}(Y_{jl}^{\top}Z_{j})^{2}\bigg\rvert \{Z_{j}\}_{j=1}^{\ns}}$. 
	 Since the variance and the variance parameter are equal, we have
	\[\mu^{\prime}=\frac{k-1}{\nm\ns}\sum_{j=1}^{\ns}\|Z_{j}\|_{2}^{2}.\]
  From Lemma \ref{lem_sbx_tail} we have that for $t>0$,
	\begin{align*}
&\lefteqn{		\bPr{\frac{1}{\ns}\sum_{j=1}^{\ns}\sum_{l=1}^{k-1}(Y_{jl}^{\top}Z_{j})^{2}\le\mu-t\bigg\rvert \{Z_{j}\}_{j=1}^{\ns}}}
\nonumber		
\\
&=	\bPr{\frac{1}{\ns}\sum_{j=1}^{\ns}\sum_{l=1}^{k-1}(Y_{jl}^{\top}Z_{j})^{2}-\mu^{\prime}\le\mu-t-\mu^{\prime}\bigg\rvert \{Z_{j}\}_{j=1}^{\ns}}\\
		&\le \exp\bigg(-\text{min}\bigg\{\frac{\nm^{2}\ns^{2}(t^{\prime})^{2}}{c_1(\sps-1)\sum_{j=1}^{\ns}\|Z_{j}\|_{2}^{4}},\frac{\nm\ns t^{\prime}}{c_2\text{max}_{j\in[\ns]}\|Z_{j}\|_{2}^{2}}\bigg\}\bigg)\numberthis\label{ip}
	\end{align*}
	where $t^\prime=\mu^\prime+t-\mu$. We now handle the $Z_{j}$-dependent terms in the exponent. In particular, we require upper bounds on the terms in the denominator and a lower bound on $\mu^{\prime}$ that hold with high probability. Recall that from Lemma \ref{lem_z4right}, we have
	\begin{align*}
		\bPr{\frac{1}{\ns}\sum_{j=1}^{\ns}\|Z_{j}\|_{2}^{4}
		\le
		\text{max}\bigg\{\bigg(1+\sqrt{\frac{c_1}{\nm}\log\frac{\ns}{\varepsilon}}\bigg)^{2},\bigg(1+{\frac{c_2}{\nm}\log\frac{\ns}{\varepsilon}}\bigg)^{2}\bigg\}}\ge 1-\varepsilon.
	\end{align*}
	
	Also, from Lemma \ref{lem_z2leftright}, we have that
	\begin{equation*}
		\bPr{\frac{1}{\ns}\sum_{j=1}^{\ns}\|Z_{j}\|_{2}^{2}\ge 1-\text{max}\bigg\{\sqrt{\frac{c_1}{\nm\ns}\log\frac{1}{\varepsilon}},\frac{c_2}{\nm\ns}\log\frac{1}{\varepsilon}\bigg\}}
		\ge 1-\varepsilon.
	\end{equation*}
	Finally, by independence of $Z_{j}$,
	\begin{align*}
		\bPr{{\text{max}_{j\in[\ns]}\|Z_{j}\|_{2}^{2}\le\mu + t}}
		&=\prod_{j=1}^{\ns}\bPr{\|Z_{j}\|_{2}^{2}\le\mu + t}\\
		&\ge \bigg(1-\exp\bigg(-\text{min}\bigg\{\frac{\nm(\mu+t-1)^{2}}{c_1},\frac{\nm(\mu+t-1)}{c_2}\bigg\}\bigg)\bigg)^{\ns}\\
		&\ge 1-\ns\exp\bigg(-\text{min}\bigg\{\frac{\nm(\mu+t-1)^{2}}{c_1},\frac{\nm(\mu+t-1)}{c_2}\bigg\}\bigg),
	\end{align*}
	which gives
	\begin{equation*}
		\bPr{\text{max}_{j\in[\ns]}\|Z_{j}\|_{2}^{2}\le 1+\text{max}\bigg\{\sqrt{\frac{c_1}{\nm}\log\frac{\ns}{\varepsilon}},\frac{c_2}{\nm}\log\frac{\ns}{\varepsilon}\bigg\}}
		\ge 1-\varepsilon.
	\end{equation*}
	Using these results together with \eqref{ip}, we have
	\begin{align}\label{ip_uncond}
		\bPr{{\frac{1}{\ns}\sum_{j=1}^{\ns}\sum_{l=1}^{k-1}(Y_{jl}^{\top}Z_{j})^{2}\le\mu-t}}
		\le & 
		\exp\bigg(-\text{min}\bigg\{\frac{\nm^{2}\ns(\frac{\sps-1}{\nm}\beta_{1}+t-\mu)^{2}}{c_1(\sps-1)\beta_{2}},
		\frac{\nm\ns (\frac{\sps-1}{\nm}\beta_{1}+t-\mu)}{c_2\beta_{3}}\bigg\}\bigg)\nonumber\\
		&+\frac{3\varepsilon}{4},\numberthis
	\end{align}
	where
	\begin{align*}
		\beta_1&=1-\text{max}\bigg\{\sqrt{\frac{c_1}{\nm\ns}\log\frac{4}{\varepsilon}},\frac{c_2}{\nm\ns}\log\frac{4}{\varepsilon}\bigg\},\\
		\beta_{2}&=\text{max}\bigg\{\bigg(1+\sqrt{\frac{c_1}{\nm}\log\frac{4\ns}{\varepsilon}}\bigg)^{2},\bigg(1+{\frac{c_2}{\nm}\log\frac{4\ns}{\varepsilon}}\bigg)^{2}\bigg\},\\
\beta_{3}&=1+\text{max}\bigg\{\sqrt{\frac{c_1}{\nm}\log\frac{4\ns}{\varepsilon}},\frac{c_2}{\nm}\log\frac{4\ns}{\varepsilon}\bigg\}.
	\end{align*}
Now, the first term on the right side of \eqref{ip_uncond} equals $\varepsilon/4$ if
\begin{align}\label{ip_conc_left_beta}
\mu-t=\frac{k-1}{m}\beta_{1}-\text{max}\bigg\{\sqrt{\frac{c_{1}\beta_{2}(k-1)}{m^{2}n}}\log\frac{4}{\varepsilon},\frac{c_{2}\beta_{3}}{mn}\log\frac{4}{\varepsilon}\bigg\}.
\end{align}
The expression above can be simplified under some mild assumptions on $n$. In particular, when $mn\ge(c_{2}^{2}/{c_{1}})\log (4/\varepsilon)$ and $m\ge(c_{2}^{2}/{c_{1}})\log(4n/\varepsilon)$, then \eqref{ip_conc_left_beta} simplifies to
\begin{align*}
\mu-t 
= \frac{\sps-1}{\nm}\bigg(1-\sqrt{\frac{c_1}{\nm\ns}\log\frac{4}{\varepsilon}}\bigg)
-\sqrt{\frac{1}{\nm\ns}\log\frac{4}{\varepsilon}}\text{max}\bigg\{\sqrt{c_{1}\frac{k-1}{m}},c_{2}\bigg\}\bigg(1+\sqrt{\frac{c_{1}}{\nm}\log\frac{4\ns}{\varepsilon}}\bigg).
\end{align*}
On the other hand, when $mn\ge(c_{2}^{2}/{c_{1}})\log 4/\varepsilon$ and $m<(c_{2}/\sqrt{c_{1}})\log(4n/\varepsilon)$, we have
\begin{align*}
\mu-t 
= \frac{\sps-1}{\nm}\bigg(1-\sqrt{\frac{c_1}{\nm\ns}\log\frac{4}{\varepsilon}}\bigg)
-\sqrt{\frac{1}{\nm\ns}\log\frac{4}{\varepsilon}}\text{max}\bigg\{\sqrt{c_{1}\frac{k-1}{m}},c_{2}\bigg\}\bigg(1+{\frac{c_{2}}{\nm}\log\frac{4\ns}{\varepsilon}}\bigg),
\end{align*}
which gives us the following simplified version of \eqref{ip_conc_left_beta} when $n\ge (c_{2}^{2}/{c_{1}})\log(4/\varepsilon)$:
\begin{align*}
		\mu-t 
		=&~ \frac{\sps-1}{\nm}\bigg(1-\sqrt{\frac{c_1}{\nm\ns}\log\frac{4}{\varepsilon}}\bigg)\nonumber\\
		&-\sqrt{\frac{1}{\nm\ns}\log\frac{4}{\varepsilon}}~\text{max}\bigg\{\sqrt{
		c_{1}\frac{k-1}{m}},c_{2}\bigg\}\text{max}\bigg\{\bigg(1+\sqrt{\frac{c_{1}}{\nm}\log\frac{4\ns}{\varepsilon}}\bigg),\bigg(1+{\frac{c_{2}}{\nm}\log\frac{4\ns}{\varepsilon}}\bigg)\bigg\}.
	\end{align*}
	This completes the proof.
\end{proof}	
\end{onlyarxiv}
\begin{lemma-app}\label{lem_ip_conc_2}
	Let $Z_{j},Y_{j1},\ldots,Y_{j,k-1}$, $j\in[\ns]$, be independent, mean-zero random vectors in $\mR^\nm$ with independent strictly subgaussian entries with variance $1/\nm$. Let $\mu=\bEE{\frac{1}{\ns}\sum_{j=1}^{\ns}\sum_{l=1}^{k-1}(Y_{jl}^{\top}Z_{j})^{2}}$. Then, there exist absolute constants $c_1$ and $c_2$ such that for any $\varepsilon>0$,
	\begin{equation*}
	\bPr{{\frac{1}{\ns}\sum_{j=1}^{\ns}\sum_{l=1}^{k-1}(Y_{jl}^{\top}Z_{j})^{2}\ge\mu+t}}\le\varepsilon,
	\end{equation*}
	for
	 	\begin{align*}
	\mu+t=&\frac{\sps-1}{\nm}\bigg(1+\sqrt{\frac{c_1}{\nm\ns}\log\frac{4}{\varepsilon}}\bigg)\\
	&+ \sqrt{\frac{1}{\nm\ns}\log\frac{4}{\varepsilon}}~\text{max}\bigg\{\sqrt{
		c_{1}\frac{k-1}{m}},c_{2}\bigg\}\text{max}\bigg\{\bigg(1+\sqrt{\frac{c_1}{\nm}\log\frac{4\ns}{\varepsilon}}\bigg),\bigg(1+{\frac{c_2}{\nm}\
		\frac{4\ns}{\varepsilon}}\bigg)\bigg\},
	\end{align*}
\end{lemma-app}
when $n\ge(c_{2}^{2}/{c_{1}})\log (4/\varepsilon)$.

\begin{onlyarxiv}
\begin{proof}
The proof is similar to that of Lemma \ref{lem_ip_conc}. We start by noting that conditioned on $\{Z_{j}\}_{j=1}^{\ns}$, the normalized sum  $(1/\ns)\sum_{j=1}^{\ns}\sum_{l=1}^{k-1}(Y_{jl}^{\top}Z_{j})^{2}$ is subexponential with parameters $v^{2}$ and $b$ where
\begin{align*}
v^{2}=\frac{c_{1}}{\nm^2\ns^2}(k-1)\sum_{j=1}^{\ns}\|Z_{j}\|_{2}^{4},\quad b=\frac{c_2}{\nm\ns}\text{max}_{j\in[\ns]}\|Z_{j}\|_{2}^{2}.
\end{align*}
Again, using the tail bound for subexponential random variables, we get
\begin{align*}
\bPr{{\frac{1}{\ns}\sum_{j=1}^{\ns}\sum_{l=1}^{k-1}(Y_{jl}^{\top}Z_{j})^{2}\ge\mu+t\bigg\rvert \{Z_{j}\}_{j=1}^{\ns}}}
&=\bPr{{\frac{1}{\ns}\sum_{j=1}^{\ns}\sum_{l=1}^{k-1}(Y_{jl}^{\top}Z_{j})^{2}-\mu^{\prime}\ge\mu+t-\mu^{\prime}\bigg\rvert \{Z_{j}\}_{j=1}^{\ns}}}\\
&\le \exp\bigg(-\text{min}\bigg\{\frac{\nm^{2}\ns^{2}(t^{\prime})^{2}}{c_1(\sps-1)\sum_{j=1}^{\ns}\|Z_{j}\|_{2}^{4}},\frac{\nm\ns t^{\prime}}{c_2\text{max}_{j\in[\ns]}\|Z_{j}\|_{2}^{2}}\bigg\}\bigg)\numberthis\label{ip_2}
\end{align*}
where
 \[\mu^\prime=\bEE{\frac{1}{\ns}\sum_{j=1}^{\ns}\sum_{l=1}^{k-1}(Y_{jl}^{\top}Z_{j})^{2}\bigg\rvert \{Z_{j}\}_{j=1}^{\ns}}=\frac{k-1}{\nm\ns}\sum_{j=1}^{\ns}\|Z_{j}\|_{2}^{2},\]
and $t^\prime=\mu+t-\mu^{\prime}$. To handle the $Z_{j}$-dependent terms in the exponent, we require high probability upper bounds on the terms in the denominator and on $\mu^{\prime}$. 
 Proceeding as in the proof of Lemma \ref{lem_ip_conc}, we have the following bounds on the terms in the denominator in \eqref{ip_2}:
 	\begin{align*}
 \bPr{{\frac{1}{\ns}\sum_{j=1}^{\ns}\|Z_{j}\|_{2}^{4}
 \le
 \text{max}\bigg\{\bigg(1+\sqrt{\frac{c_1}{\nm}\log\frac{\ns}{\varepsilon}}\bigg)^{2},\bigg(1+{\frac{c_2}{\nm}\log\frac{\ns}{\varepsilon}}\bigg)^{2}\bigg\}}}\ge 1-\varepsilon.
 \end{align*}
and
 \begin{equation}\label{conc_z2max}
 \bPr{\text{max}_{j\in[\ns]}\|Z_{j}\|_{2}^{2}\le 1+\text{max}\bigg\{\sqrt{\frac{c_1}{\nm}\log\frac{\ns}{\varepsilon}},\frac{c_2}{\nm}\log\frac{\ns}{\varepsilon}\bigg\}}
 \ge 1-\varepsilon.
 \end{equation}
 Also, from Lemma \ref{lem_z2leftright},
 \begin{align}\label{conc_z2right}
 \bPr{\frac{1}{\ns}\sum_{j=1}^{\ns}\|Z_{j}\|_{2}^{2}\le
  1+ \text{max}\bigg\{\sqrt{\frac{c_1}{\nm\ns}\log\frac{1}{\varepsilon}},\frac{c_2}{\nm\ns}\log\frac{1}{\varepsilon}\bigg\}}
 \ge 1-\varepsilon.
 \end{align}
 We note that although a high probability upper bound on $\text{max}_{j\in[\ns]}\|Z_{j}\|_{2}^{2}$ implies a high probability upper bound on  $(1/\ns)\sum_{j=1}^{\ns}\|Z_{j}\|_{2}^{2}$, we specifically use the bound in \eqref{conc_z2right} since the deviation term has better dependence on $n$ (which is lost in \eqref{conc_z2max} due to a union bound step). A $\sqrt{(1/m)\log(n/\varepsilon)}$ or $(1/m)\log(n/\varepsilon)$ type dependence, on the other hand, would lead to constraints on $m$.
 
 Using these results along with \eqref{ip_2}, we have
 \begin{align}\label{ip_2_uncond}
 \bPr{{\frac{1}{\ns}\sum_{j=1}^{\ns}\sum_{l=1}^{k-1}(Y_{jl}^{\top}Z_{j})^{2}\ge\mu+t}}
 \le & 
 \exp\bigg(-\text{min}\bigg\{\frac{\nm^{2}\ns(\mu+t-\frac{\sps-1}{\nm}\beta_{1})^{2}}{c_1(\sps-1)\beta_{2}},
 \frac{\nm\ns (\mu+t-\frac{\sps-1}{\nm}\beta_{1})}{c_2\beta_{3}}\bigg\}\bigg)\nonumber\\
 &+\frac{3\varepsilon}{4},\numberthis
 \end{align}
 where
\begin{align*}
\beta_{1}=1+ \text{max}\bigg\{\sqrt{\frac{c_1}{\nm\ns}\log\frac{4}{\varepsilon}},\frac{c_2}{\nm\ns}\log\frac{4}{\varepsilon}\bigg\},
\end{align*}
 \begin{align*}
 \beta_{2}&=\text{max}\bigg\{\bigg(1+\sqrt{\frac{c_1}{\nm}\log\frac{4\ns}{\varepsilon}}\bigg)^{2},\bigg(1+{\frac{c_2}{\nm}\log\frac{4\ns}{\varepsilon}}\bigg)^{2}\bigg\},
 \end{align*}
 and
   \begin{align*}
 \beta_{3}=1+\text{max}\bigg\{\sqrt{\frac{c_1}{\nm}\log\frac{4\ns}{\varepsilon}},\frac{c_2}{\nm}\log\frac{4\ns}{\varepsilon}\bigg\}=\sqrt{\beta_{2}}.
 \end{align*}
 Simplifying as we did in Lemma \ref{lem_ip_conc} under the assumption that $\ns\ge(c_{2}^{2}/{c_{1}})\log (4/\varepsilon)$, we see that if
 	\begin{align*}
 		\mu+t=&\frac{\sps-1}{\nm}\bigg(1+\sqrt{\frac{c_1}{\nm\ns}\log\frac{4}{\varepsilon}}\bigg)\\
 	&+ \sqrt{\frac{1}{\nm\ns}\log\frac{4}{\varepsilon}}~\text{max}\bigg\{\sqrt{
 		c_{1}\frac{k-1}{m}},c_{2}\bigg\}\text{max}\bigg\{\bigg(1+\sqrt{\frac{c_1}{\nm}\log\frac{4\ns}{\varepsilon}}\bigg),\bigg(1+{\frac{c_2}{\nm}\
 			\frac{4\ns}{\varepsilon}}\bigg)\bigg\},
 	\end{align*}
 	then the first term on the right side of \eqref{ip_2_uncond} is less than $\varepsilon/4$, which completes the proof.
\end{proof}
\end{onlyarxiv}
\section{Fourth moment of the minimum eigenvalue of a Wishart matrix}\label{app_smin}

\begin{lemma-app}\label{lem_smin}
	Let $\bphi\in\mR^{\nm\times\sps}$ with independent $\cN(0,1)$ entries and let $A=\bphi\bphi^{\top}$. If $Z$ denotes the minimum eigenvalue of $A$, then for $\sps-\nm>7$, 
	\begin{equation*}
	\bEE {Z^{-4}}\le\frac{c}{\sps^{4}(1-\nm/\sps)^{8}}.
	\end{equation*}
\end{lemma-app}
\begin{proof}
	Since $Z$ is a nonnegative random variable, we have for $\theta>0$,
	\begin{align*}
	\bEE {Z^{-4}}
	&\le
	\tht+\int_{\tht}^{\infty}\bPr{Z\le u^{-\frac{1}{4}}}du\\
	&\le \tht+\frac{8}{\sps^{4}}\int_{0}^{\frac{\tht^{-\frac{1}{8}}}{\sqrt{\sps}}}\bPr{Z\le
		k\ep^{2}}\frac{1}{\ep^{9}}d\ep,
	\end{align*} 
	where we used $u^{-\frac{1}{4}}=\sps\ep^{2}.$
	The density of the smallest eigenvalue of a Wishart matrix  with parameters $k$ and $m$ ($A$ in this case) is
	known in closed form \cite [Lemma 4.1]{Chen_SIAM_2005}, which we restate here:
	\begin{align*}
	\bPr{Z\le \sps\ep^{2}}&\le\frac{1}{\Gamma(\sps-\nm+2)}(\ep
	\sps)^{\sps-\nm+1}\\ &\le
	\bigg(\frac{e}{\sps-\nm+1}\bigg)^{\sps-\nm+1}(\ep
	\sps)^{\sps-\nm+1},
	\end{align*}
	where $\Gamma(\cdot)$ denotes the gamma function and $\Gamma(n)=(n-1)!$ for integer $n$. Using this, we get
	\begin{align*}
	\se Z^{-4} &\le \tht
	+\frac{8}{\sps^{4}}\bigg(\frac{e\sps}{\sps-\nm+1}\bigg)^{\sps-\nm+1}\int_{0}^{\frac{\tht^{-\frac{1}{8}}}{\sqrt{\sps}}}(\ep)^{\sps-\nm-8}d\ep\\ &=\tht+\frac{8}{\sps^{4}}\bigg(\frac{e\sps}{\sps-\nm+1}\bigg)^{\sps-\nm+1}
	\frac{1}{\sps-\nm-7}\bigg(\frac{\tht^{-\frac{1}{4}}}{\sps}\bigg)^{\frac{\sps-\nm-7}{2}}.
	\end{align*}
	Choosing $\tht=\bigg(\frac{e\sqrt{\sps}}{\sps-\nm+1}\bigg)^8$
	and simplifying gives
	\begin{align*}
	\bEE {Z^{-4}} \le
	\frac{9e^8\sps^4}{(\sps-\nm-7)^8}\le \frac{c}{\sps^4\bigg(1-\frac{\nm}{\sps}\bigg)^8}.
	\end{align*}
\end{proof}


\section{{Toolbox: Some useful concentration bounds}}\label{app_prelim}
\begin{onlyjournal}
Proofs of the results in this section can be found in \cite{Ramesh_arxiv_2019}.
\end{onlyjournal}

\begin{lemma-app}\label{lem_sbx_tail}
	Let $X$ be a subexponential random variable with parameters $v^{2}$ and $b>0$ (denoted $X\sim\sbx(v^{2},b)$), that is, 
	\begin{equation*}
	\bEE{\exp(\theta(X-\bEE{X}))}
	\le\exp\bigg(\frac{\theta^{2}v^{2}}{2}\bigg),\quad |\theta|<\frac{1}{b}.
	\end{equation*}	
	Then,
	\begin{equation*}
	\bPr{|X-\bEE{X}|\ge t}\le 2\exp\bigg(-\text{min}\bigg\{\frac{t^{2}}{2v^{2}},\frac{t}{2b}\bigg\}\bigg).
	\end{equation*}
\end{lemma-app}

\begin{onlyarxiv}
\begin{proof}
	See \cite[Proposition~ 2.2]{Wainwright_text_2019}.
\end{proof}
\end{onlyarxiv}

\begin{lemma-app}\label{lem_sbx_sbg}
	Let $X\sim\sbg(\sigma^2)$ with $\bEE X=0$. Then $X^2\sim\sbx(128\sigma^4,8\sigma^2)$.
\end{lemma-app}

\begin{onlyarxiv}
\begin{proof}
	Let $Y=X^2$. We start by upper bounding the moment generating function (MGF) of $Y$. For $\theta>0$,
	\begin{align*}
	\bEE{e^{\theta(Y-\bEE Y)}}
	&=\bEE{\sum_{q=0}^{\infty}\frac{(\theta(Y-\bEE Y))^q}{q!}}\\
	&\le 1+\sum_{q=2}^{\infty}\frac{(2\theta)^{q}\bEE{Y^{q}}}{q!}\\
	&=1+\sum_{q=2}^{\infty}\frac{(2\theta)^q}{q!}\bEE{X^{2q}},
	\end{align*}
	where in the second step we used 
	$(\bEE{|Y-\bEE{Y}|^{q}})^{\frac{1}{q}}\le
	(\bEE{|Y|^{q}})^{\frac{1}{q}}+\mu\le 2(\bEE{Y^{q}})^{\frac{1}{q}}$.
	
	Now, for $X\sim\sbg(\sigma^2)$, we have the following upper bound on the moments of $X$ from \cite[Theorem 2.1]{BLM_conc_text} :
	$\bEE{X^{2q}}\le 2q! 2^q \sigma^{2q}$.
	This gives
	\begin{align*}
	\bEE{e^{\theta(Y-\bEE Y)}}&\le 1+2\sum_{q=2}^{\infty}\frac{\theta^q}{q!}q! 2^{2q}
	\sigma^{2q}\\
	&=1+\frac{32\theta^2\sigma^4}{1-4\theta\sigma^2},\quad \theta<\frac{1}{4\sigma^{2}}.
	\end{align*}
	For $\theta\le1/8\sigma^2$, we get
	\begin{align*}
	\bEE{e^{\theta(Y-\bEE{Y})}}\le 1+64\theta^2\sigma^4
	\le e^{64\theta^2\sigma^4},
	\end{align*}
	that is, $Y\sim\sbx(128\sigma^4,8\sigma^2)$.
\end{proof}
\end{onlyarxiv}

\begin{lemma-app}\label{lem_sbx_txf}
	Let $X_i\sim\sbx(v_i^2,b_i)$ be independent subexponential random variables for $i\in[n]$. Then, for a constant $a\in\mR$, we have that $aX_1\sim\sbx(a^2v_1^2,|a|b_1)$ and $\sum_{i=1}^{n}X_i\sim\sbx(\sum_{i=1}^{n}v_{i}^{2},\text{max}_{i\in[n]}b_i)$.
\end{lemma-app}

\begin{onlyarxiv}
\begin{proof}
	The proof involves bounding the MGF of the transformed random variables and noting that it has the same form as the MGF of a subexponential random variable with the parameters appropriately transformed. Specifically, for $0<\theta<1/|a|b_{1}$, we have
	\begin{align*}
	\bEE{\exp(a\theta(X_{1}-\bEE{X_{1}}))}
	\le \exp\bigg(\frac{a^{2}\theta^{2}v_{1}^{2}}{2}\bigg),
	\end{align*}
	that is, $aX_{1}\sim\text{subexp}(a^{2}v_{1}^{2},|a|b_{1})$.
	
	Similarly, bounding the MGF of the sum $Y=\sum_{i=1}^{n}X_i$, we get
	\begin{align*}
	\bEE{\exp(\theta(Y-\bEE Y))}
	&=\prod_{i=1}^{n}\bEE{\exp(\theta(X_{i}-\bEE{X_{i}} ))}\\
	&\le\prod_{i=1}^{n}\exp\bigg(\frac{\theta^{2}v_{i}^{2}}{2}\bigg),
	\end{align*}
	when $|\theta|<1/b_{i}$ for all $i\in[n]$. That is, for  $|\theta|<1/(\text{max}_{i\in [n]}b_{i})$,
	\begin{align*}
		\bEE{\exp(\theta(Y-\bEE Y))}
		\le \exp\bigg(\theta^{2}\sum_{i=1}^{n}\frac{v_{i}^{2}}{2}\bigg)
	\end{align*}
	which shows that $Y\sim\text{subexp}(\sum_{i=1}^{n}v_{i}^{2},\text{max}_{i\in[n]}b_{i})$.
\end{proof}
\end{onlyarxiv}

\begin{lemma-app}\label{lem_moments}
	Let $W$ and $Z$ be $\nm$-dimensional random vectors having independent zero-mean subgaussian entries with variance $1/\nm$ and fourth moment $3/\nm^2$.
	Then,
	\begin{equation*}
	\bEE{\|Z\|_{2}^{4}}=1+\frac{2}{\nm},\quad \text{and}\quad 
	\bEE{(Z^\top W)^2}=\frac{1}{\nm}.
	\end{equation*}
\end{lemma-app}
\begin{onlyarxiv}
\begin{proof}
	The proof is based on a straightforward calculation. We have
	\begin{align*}
	\bEE{\|Z\|_{2}^{4}}
	&=\sum_{i=1}^{\nm}\bEE{Z_{i}^4}+\sum_{i\ne j}\bEE{Z_{i}^{2}Z_{j}^{2}}\\
	&=1+\frac{2}{\nm},
	\end{align*}
	and
	\begin{align*}
	\bEE{(Z^\top W)^2}
	&=\se\bigg[ \se\bigg[\bigg(\sum_{i=1}^{\nm}Z_{i}^{2}W_{i}^{2}+\sum_{i\ne j}Z_{i}W_{i}Z_{j}W_{j}\bigg)\bigg\rvert Z\bigg]\bigg]\\
	&=\frac{1}{\nm}.
	\end{align*}
\end{proof}
\end{onlyarxiv}

\begin{lemma-app}\label{lem_traceAB}
	Let $A, B\in\mR^{m\times m}$ be symmetric, positive definite matrices and let $a_1\ge\cdots\ge a_m$ and $b_1\ge\cdots\ge b_m$ denote their respective ordered eigenvalues. Then,
	\begin{equation*}
		\text{Tr}(AB)\le \sum_{i=1}^{m}a_i b_i.
	\end{equation*}
\end{lemma-app}
\begin{onlyarxiv}
\begin{proof}
The proof follows from  \cite[Theorem 20.B.1]{Marshall_text_2011}.
\end{proof}
\end{onlyarxiv}
\end{appendices}

\bibliography{IEEEabrv,bibfile}
\bibliographystyle{IEEEtranS}

\end{document}